\documentclass[showkeys,twocolumn,superscriptaddress]{revtex4-2}
\usepackage{ulem}
\usepackage{times}
\usepackage{graphicx}
\usepackage[pdftex,dvipsnames]{xcolor}
\usepackage{amsmath}
\usepackage{amssymb}
\usepackage{amsfonts}
\usepackage{braket}
\usepackage{mhchem}
\usepackage{graphicx}
\usepackage{enumitem}
\usepackage{amsthm}
\usepackage{amsfonts}
\usepackage{mathtools}
\usepackage{float}
\usepackage{todonotes}

\newtheorem{lemma}{Lemma}[section]

\newcommand{\braketm}[2]{\langle #1|#2\rangle}
\newcommand{\ketbra}[2]{|#1\rangle\langle #2|}

\begin{document}

\title{Long-range photonic device-independent quantum key distribution\\ using SPDC sources and linear optics}
\author{Morteza Moradi}
\affiliation{Institute of Informatics, Faculty of Mathematics, Informatics and Mechanics, University of Warsaw, Banacha 2c, 02--097 Warsaw, Poland}
\author{Maryam Afsary}
\affiliation{Institute of Informatics, Faculty of Mathematics, Informatics and Mechanics, University of Warsaw, Banacha 2c, 02--097 Warsaw, Poland}
\author{Piotr Mironowicz}
\affiliation{Center for Theoretical Physics, Polish Academy of Sciences, Aleja Lotników 32/46, 02-668 Warsaw, Poland}
\affiliation{Faculty of Electronics, Telecommunications and Informatics, Gdańsk University of Technology, Narutowicza 11/12, 80-233 Gdańsk, Poland}
\author{Enky Oudot}
\affiliation{ICFO - Institut de Ciencies Fotoniques, The Barcelona Institute
of Science and Technology, 08860 Castelldefels, Barcelona, Spain}
\affiliation{LIP6, CNRS, Sorbonne Universit\'e, 4 place Jussieu, F-75005 Paris, France}
\author{Magdalena Stobi\'nska-Moretto}
\affiliation{Center for Hybrid Quantum-Classical Information Technologies ``QLAB'', University of Warsaw, Pasteura 5, 02--093 Warsaw, Poland}

\begin{abstract}
 We address the question of the implementation of long-distance device-independent quantum key distribution (DI QKD) by proposing two experimentally viable schemes. Those schemes only use  spontaneous parametric down-conversion (SPDC) sources and linear optics. They achieve favorable key rate scaling proportional to the square root of channel transmittance $\eta_t$, matching the twin-field protocol advantage. We demonstrate positive asymptotic key rates at detector efficiencies as low as 80\%, bringing DI QKD within the reach of current superconducting detector technology. Our security analysis employs the Entropy Accumulation Theorem to establish rigorous finite-size bounds, achieving finite-key rates at a detector efficiency of 90\%. This work represents a critical milestone toward device-independent security in quantum communication networks, providing experimentalists with practical implementation pathways while maintaining the strongest possible security guarantees against quantum adversaries.
\end{abstract}

\maketitle

\noindent
\textit{Introduction.---} Quantum key distribution (QKD) promises information-theoretic security for communications, leveraging quantum mechanics to detect eavesdropping attempts. Among various protocols, device-independent (DI) schemes offer the highest level of security by eliminating the need to trust the devices~\cite{Zapatero2023,Primaatmaja2023}. This remarkable property addresses a critical vulnerability: in practice, imperfect or compromised devices can leak information, as demonstrated by successful attacks on QKD systems~\cite{Gerhardt2011,Lydersen2010,Zhao2008,Weier2011}.

The security of DI QKD is based on a profound insight. The violation of Bell inequality certifies the presence of quantum nonlocality that no classical system can reproduce. By monitoring this violation, users can bound the information accessible to any eavesdropper. The security proof is based solely on observed measurement statistics.

Despite its elegance, implementing DI QKD over long distances faces the challenge of closing the detection loophole while maintaining practical key generation rates. This loophole arises when measurement inefficiencies allow local hidden-variable models to reproduce the observed correlations. For photonic implementations, detector efficiencies exceeding 82.8\% are required for Bell tests with maximally entangled states, although this threshold can be lowered to 66.7\% for other states and specific measurements~\cite{Giustina2013,Morteza}. Moreover, in most QKD protocols key rates decay exponentially with distance due to channel losses, limiting practical distance to 100--150~km in standard fibers~\cite{Pirandola2020}.

Recent experiments demonstrated DI QKD using matter qubits with near-perfect detection efficiency~\cite{NadlingerNature2022,Weinfurter-DIQKD}. Although they validated the principle, they underscored the need for fully photonic solutions compatible with optical fiber network infrastructure. Initial steps have been reported in Ref.~\cite{Pan-DIQKD}, yet achieving long-distance implementation remains an open challenge. In particular, the quantum repeater approach faces fundamental constraints in DI scenarios~\cite{Sadhu2023}. 

A promising idea to overcome distance limitations in QKD is the twin-field (TF) configuration~\cite{Lucamarini2018}. This architecture enables key rates that scale with the square root of channel transmittance, a significant improvement over the linear scaling of conventional QKD. By combining the security of DI QKD with the scaling of TF protocols, DI-TF QKD could enable secure communication over long distances without trusting any network nodes. Such a fully photonic DI-TF QKD scheme has recently been proposed and analyzed with respect to security in Ref.~\cite{Acin2024}, but the setup required quantum dot sources and a nonlinear measurement scheme that poses experimental challenges.

\begin{figure*}[t]\centering
	\raisebox{3.2cm}{(a)}\kern0.1cm\includegraphics[height=3.5cm]{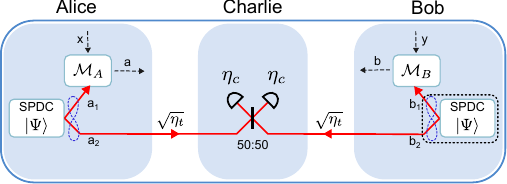}\hfil
	\raisebox{3.2cm}{(b)}\kern0cm\includegraphics[height=3.5cm]{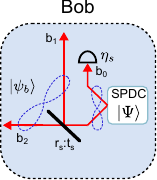}\hfil
	\raisebox{3.2cm}{(c)}\kern0cm\includegraphics[height=3.5cm]{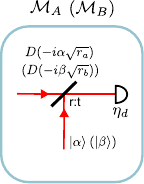}
  \caption{Long-range DI-TF QKD based on heralded entanglement distribution.
  (a) In the 1-photon protocol, Alice and Bob use spontaneous parametric down-conversion (SPDC) sources pumped by pulsed lasers to generate local bipartite multi-photon entanglement -- two-mode squeezed vacuum states (TMSV). The idler modes $a_2$ and $b_2$ travel through lossy channels with transmittance $\sqrt{\eta_t}$ to reach Charlie's central station, which comprises a symmetric beam splitter and two detectors with efficiency $\eta_c$. The events where exactly one of the two detectors registers a photon herald the state in Eq.~\eqref{eq:psiout}.
  (b) In the 2-photon protocol, Bob employs a single-photon source obtained by heralding one output from an SPDC crystal at a detector of efficiency $\eta_s$ with the other output routed through a beam splitter with reflectivity-to-transmissivity ratio $r_s:t_s$, effectively generating local single-photon entanglement $\ket{\psi_b}$.
  (c) The heralded state, Eqs.~\eqref{eq:psiout} or \eqref{eq:psiout2ph}, enters Alice's and Bob's measurement systems $\mathcal{M}_A$ and $\mathcal{M}_B$. Each measurement comprises a displacement operation, implemented by interfering the signal $a_1$ ($b_1$) with a coherent state $\ket{\alpha}$ ($\ket{\beta}$) on a beam splitter with transmissivity $t_{a,b}\approx 1$, followed by a detector with efficiency $\eta_d$. All detectors used discriminate only between `photons' and `no-photons' cases (on/off detection).}
  \label{fig:setup}
\end{figure*}

In this work, we study two DI-TF QKD protocols based on heralded entanglement distribution using spontaneous parametric down-conversion (SPDC) sources, linear optics, and on/off single-photon detectors. The first scheme, demonstrated experimentally in Ref.~\cite{Caspar2020} and further analyzed for Bell experiments in Ref.~\cite{PRA}, produces a maximally entangled single-photon state $\ket{\Psi_\text{out}^\text{(1ph)}}=\frac{1}{\sqrt{2}}\left(\ket{01}\pm i\ket{10}\right)$. We provide a rigorous security proof for a DI QKD protocol using this state; however, the measurement scheme is suboptimal for $\ket{\Psi_\text{out}^\text{(1ph)}}$, resulting in a high detection-efficiency threshold. We refer to it as the \textit{1-photon protocol}.

The second scheme constitutes our main contribution and significantly relaxes experimental requirements: for the first time, we achieve threshold efficiencies compatible with state-of-the-art experiments together with a rigorous security proof. The key ingredient is to herald states of the form $\ket{\Psi_\text{out}^\text{(2ph)}}=\ket{00}+\epsilon\ket{11}$; we call this the \textit{2-photon protocol}. The state $\ket{\Psi_\text{out}^\text{(2ph)}}$ offers two advantages for DI QKD. First, key generation uses $Z$-basis measurements with on/off detectors, yielding perfectly correlated outcomes; losses only weakly affect the error rate, as most events correspond to the vacuum component $\ket{00}$. Second, it allows CHSH violation down to the 66\% loss threshold. While unbalanced states are known to optimize key rates at a given efficiency~\cite{Brown2024}, the CHSH inequality is suboptimal for DI QKD with such states. We therefore develop numerical methods for rigorous security proofs against general quantum attacks, using the Entropy Accumulation Theorem (EAT)~\cite{Dupuis2019,Dupuis2020} with Bell inequalities tailored to $\ket{\Psi_\text{out}^\text{(2ph)}}$.

\bigskip
\noindent
\textit{The protocol.---} The setup we consider is depicted in Fig.~\ref{fig:setup}. Alice and Bob each generate multi-photon entangled states at their respective locations. In each round, they randomly choose one of three measurement settings: two for the nonlocality test and one for key generation. Charlie heralds long-range entanglement and communicates via a classical channel which rounds to keep.
The basic ingredients are two identical two-mode squeezed vacuum (TMSV) sources, obtained by pumping two SPDC crystals. These sources generate the state $\ket{\Psi} = \sum_n \sqrt{\lambda_n} \ket{n,n}_{1,2}$, with photon-number correlations across modes 1 and 2, where $\lambda_n = \tanh^{2n}{(g)}/\cosh^2(g)$ quantifies the probability of $n$-photon emissions in each mode, and $g$ is the parametric gain.
Since vacuum and single-photon events dominate ($\lambda_0 \gg \lambda_1 \gg \lambda_2 \cdots$), for clarity we approximate this state as $\ket{\Psi} \approx \frac{1}{\mathcal{N}} \left(\sqrt{\lambda_0}\ket{00}+ \sqrt{\lambda_1}\ket{11}\right)$; however, all computations and plots in this Letter include the full multiphoton contributions as described in the Supplemental Material (SM).

\bigskip
\noindent
\textit{The 1-photon protocol.---} In the scenario depicted in Fig.~\ref{fig:setup}a, Alice and Bob send one TMSV mode each to Charlie's central station, equidistant from both. Charlie interferes the incoming modes at a symmetric beam splitter and heralds only when exactly one of the two detectors clicks, projecting the distributed state onto high-fidelity entanglement~\cite{PRA,Caspar2020}
\begin{equation}
    \ket{\Psi_\text{out}^\text{(1ph)}} =  \frac{1}{\sqrt{2}}\left(\ket{0,1}_{a_1, b_1} \pm i \ket{1,0}_{a_1, b_1}\right).
    \label{eq:psiout}
\end{equation}
With Charlie's detector efficiency $\eta_c$, photon-number statistics $\lambda^{(a)}$ ($\lambda^{(b)}$) for Alice (Bob), the heralding probability is
$P_h^\text{(1ph)} \!\!=\!\!\left(\lambda_0^{(a)}\lambda_1^{(b)}+\lambda_0^{(b)}\lambda_1^{(a)}\right) \eta_c\sqrt{\eta_t} \!\approx \!O(\sqrt{\eta_t})$.

This protocol achieves breakthrough scaling. At distance $L$, channel transmittance is $\eta_t = 10^{-\alpha_{\text{att}}L/10}$, with attenuation coefficient $\alpha_{\text{att}} = 0.2 \; \mathrm{km}^{-1}$. Since the photons travel half distance to Charlie, the losses scale as $\sqrt{\eta_t}$, versus linear decay ($\eta_t$) in conventional QKD. This square-root improvement allows previously inaccessible distances. 

In practice, Charlie heralds a mixed state $\rho_\text{out}$ rather than the pure state Eq.~\eqref{eq:psiout}, but $\ket{\Psi_\text{out}^\text{(1ph)}}\bra{\Psi_\text{out}^\text{(1ph)}}$ remains dominant. This robustness proves critical for implementations. In our further analysis, we assume perfect heralding detectors at Charlie's ($\eta_c = 1$), as their inefficiency merely reduces protocol success rate without compromising security.

\bigskip
\noindent
\textit{The 2-photon protocol.---} In this second approach, Fig.~\ref{fig:setup}b, Alice's setup remains unchanged while Bob uses a single-photon source, implemented with a heralded SPDC~\cite{Kaneda2016} using a detector of efficiency $\eta_s$. Bob routes the photon through a beam splitter with transmissivity $t_s$, creating a local entangled state $\ket{\psi_b} = \sqrt{t_s}\ket{0,1}_{b_1,b_2}+ e^{i\phi} \,\sqrt{1-t_s}\,\ket{1,0}_{b_1,b_2}$. A single click at Charlie's station heralds an entangled state
\begin{equation}
    \ket{\Psi_\text{out}^\text{(2ph)}} \propto  \sqrt{\lambda_0^{(a)} t_s}\,\ket{0,0}_{a_1, b_1} 
    \pm e^{i\phi} \, \sqrt{\lambda_1^{(a)} (1-t_s)} \ket{1,1}_{a_1, b_1}.
    \label{eq:psiout2ph}
\end{equation}
For small $t_s$, this state takes the form of a two-mode squeezed state, which is known to exhibit strong loss robustness for Bell test with the considered measurement~\cite{Vivoli2016}.

This protocol achieves the same scaling as the 1-photon one, $P_h^\text{(2ph)} \!=\! P_s\left( \lambda_0^{(a)}t_s + \lambda_1^{(a)}(1-t_s) \right)\eta_c\sqrt{\eta_t}\approx O(\sqrt{\eta_t})$,
where $P_s$ is Bob's single-photon generation probability.
For full derivation of density matrices and performance of protocols, see SM, Section~\ref{Entang-sec}.

\bigskip
\noindent
\textit{The nonlocality test.---} To extract the secret key, Alice and Bob perform local measurements $\mathcal{M}_A$ and $\mathcal{M}_B$. Alice selects between two settings $x \in \{1,2\}$ while Bob chooses from three $y \in \{1,2, 3\}$, each yielding binary outcomes $a,b \in \{\pm 1\}$, cf.\ Ref.~\cite{Acin2007}.

For the Bell test, both parties randomly select from their first two settings and evaluate the CHSH parameter~\cite{CHSH}
\begin{equation}
    S = \langle A_1 B_1 \rangle + \langle A_1 B_2 \rangle + \langle A_2 B_1 \rangle - \langle A_2 B_2 \rangle,
\end{equation}
where $A_x$ and $B_y$ denote Alice's and Bob's measurement observables, respectively, with correlations $\langle A_x B_y \rangle = \sum_{a,b} p(a=b|x,y) - p(a \neq b |x,y)$. This parameter quantifies the observed nonlocality and bounds the information accessible to any eavesdropper, satisfying $S \leq 2$ for all local hidden-variable theories.

To violate Bell inequalities, one must perform mutually incompatible measurements. The states in Eqs.~\eqref{eq:psiout} and~\eqref{eq:psiout2ph} are encoded in the Fock basis, $\{\ket{0}, \ket{1}\}$. While measuring in the $z$ basis is straightforward using single-photon detectors, optimal Bell tests require challenging projections onto photon-number superpositions\break $\{\tfrac{1}{\sqrt{2}}(\ket{0}+\ket{1}),\tfrac{1}{\sqrt{2}}(\ket{0}-\ket{1})\}$.
To overcome this constraint, we employ a practical alternative accessible in modern laboratories, inspired by Ref.~\cite{Banaszek1999}. 

Our measurement scheme, Fig.~\ref{fig:setup}c, implements: (i) displacement operations $D(\delta) = e^{\delta a^\dagger - \delta^* a}$ with setting-dependent $\delta_a = -i\alpha\sqrt{r_a}$ and $\delta_b = -i\beta\sqrt{r_b}$, followed by (ii) `photons'/`no-photons' (on/off) detection with efficiency $\eta_d$. The effective positive operator-valued measures (POVMs) are
\begin{equation}
    \mathcal{M}_0^{(A)} = D(\delta_a) E_{0}^{\eta_d} D^{\dagger}(\delta_a),
    \quad
    \mathcal{M}_0^{(B)} = D(\delta_b) E_{0}^{\eta_d} D^{\dagger}(\delta_b),
\end{equation}
where $E_0^{\eta_d} = (1-\eta_d)^{a^\dagger a}$ represents the effective no-click operator. To evaluate the CHSH value under realistic conditions with multimode sources and detection losses, we compute
\begin{equation}
    \begin{aligned}
    p_{a}(0|\delta_{a(b)}) &{}= \mathrm{Tr} \big( \rho_{a(b)} \,\mathcal{M}_0^{A(B)} \big),\\
    p(00|\delta_a\delta_b) &{}= \mathrm{Tr} \big( \rho_{\text{out}}\, \mathcal{M}_0^{(A)} \otimes \mathcal{M}_0^{(B)}),
    \end{aligned}
\end{equation}
with $\rho_{a(b)}$ denotes Alice's (Bob's) reduced density matrix.
The detailed derivations of the marginals and correlators, including higher photon-number contributions, are presented in SM Sections~\ref{POVMs-sec} and~\ref{Distributions-sec}.

Remarkably, both protocols maintain near-maximal CHSH violations even under severe transmission losses: $S = 2.688$ for $\ket{\Psi_\text{out}^\text{(1ph)}}$ and $S=2.686$ for $\ket{\Psi_\text{out}^\text{(2ph)}}$. This robustness stems from the protocol's inherent noise tolerance: under high loss, the heralded state resides in the qubit subspace where violations are strongest. By decreasing the value of $g$, we limit the contribution of heralded multiphoton components and their impact on reduced violations. However, this affects the heralding probability $P_h$, which is crucial in the finite-size regime. Consequently, channel loss primarily affects protocol performance, while the CHSH value $S$ and the asymptotic key rate $r$ remain distance-independent.

\begin{figure}\centering
    \includegraphics[width=\columnwidth]{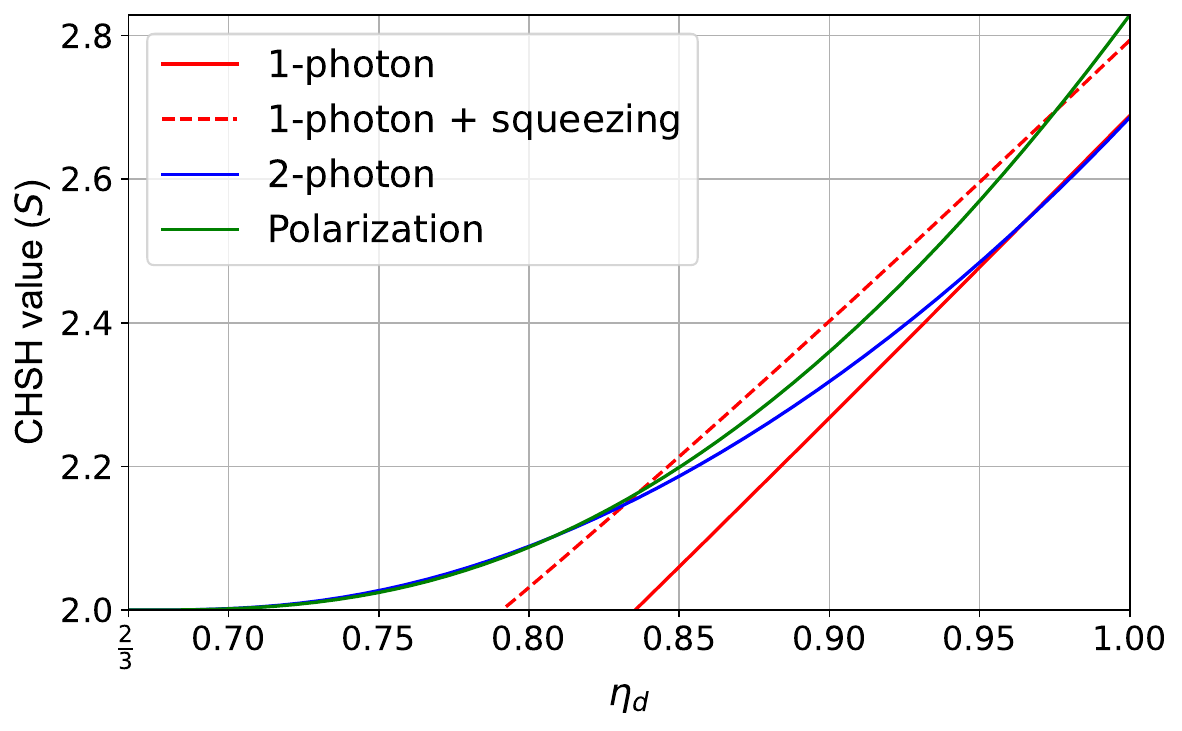}
    \caption{{Maximal CHSH parameter $S$ as a function of detector efficiency $\eta_d$ computed for: 1-photon protocol (red solid line, with dashed red line including squeezing operations in the measurements~\cite{Acin2024}), 2-photon protocol (blue) and, for reference, a polarization-based protocol~\cite{Oudot2024} (green).}}
\label{fig:Bell vs. Eta}
\end{figure}

Fig.~\ref{fig:Bell vs. Eta} illustrates how detection efficiency $\eta_d$ influences the maximal CHSH value ($S$). The state $\ket{\Psi_\text{out}^\text{(1ph)}}$ suffers losses in both superposition components. When photons are lost, the anti-correlation in this state degrades towards a full correlation $\ket{00}$, leading to a rapid drop in $S$. In Ref.~\cite{Acin2024}, an additional squeezing operation through a non-linear crystal $\chi^{(2)}$ is inserted between the displacement operation and the detector to partially mitigate this effect, but it is experimentally challenging. In contrast, state $\ket{\Psi_\text{out}^\text{(2ph)}}$ exhibits much higher loss resilience as losses only affect the $\ket{11}$ component. By tuning $t_s$ and $g$ to reduce the coefficient $\lambda_1(1-t_s)$, the protocol can approach the Eberhard limit for photon-counting measurements~\cite{Morteza}. We compare these results with the polarization-based protocol~\cite{Oudot2024}, where states of the form $\cos(\theta) \ket{HH} + \sin(\theta) \ket{VV}$ experience symmetric losses across both components. Here, however, strategic post-processing can preserve correlations: if Alice and Bob assign no-detection events the same value as $\ket{H}$-detection, the $\ket{HH}$ correlations remain unaffected by loss. Taking $\theta\to 0$ suppresses contributions from $\ket{VV}$, achieving robustness comparable to the 2-photon protocol.

\bigskip
\noindent
\textit{Security analysis in the asymptotic regime.---} The asymptotic secure key rate for DI QKD protocols, which offers resilience against collective attacks by a quantum eavesdropper (Eve)~\cite{Acin2007,Pironio2009}, is lower-bounded by the Devetak--Winter formula~\cite{DW}
\begin{equation}
    r_{\infty} \geq H(A_1|E) - H(A_1|B_3),
    \label{DW-ineq}
\end{equation}
where $H(X|Y)$ denotes the conditional von Neumann entropy. This expression quantifies Alice and Bob's information advantage over Eve: the first term captures Eve's uncertainty about Alice's outcomes, while the second represents error-correction costs. Although $H(A_1|B_3)$ follows directly from measurement statistics, bounding Eve's information $H(A_1|E)$ requires careful analysis. {To evaluate this, we} exploit the connection between CHSH inequality violations and information security. Using the analytical framework from Ref.~\cite{Ho2020} with noisy preprocessing, we obtain
\begin{equation}
    \begin{aligned}[c]
    r \geq{}& 1 - h\left(\tfrac{1+\sqrt{(S/2)^2-1}}{2}\right)- H(A_1|B_3)\\
    &{}+ h\left(\tfrac{1+\sqrt{1-q(1-q)(8-S^2)}}{2}\right),\end{aligned}
    \label{SKR-CHSH+NP}
\end{equation}
where $h(x) = -x\log_2(x) - (1-x)\log_2(1-x)$ is the binary entropy function, and $q$ denotes the bit-flip probability during noisy preprocessing, a standard method to enhance key rate (see SM, Section~\ref{LB_Rate-sec}.A).

For tighter bounds, we utilize the full probability distribution $p(a,b|x,y)$. Following the Brown–Fawzi–Fawzi (BFF) method~\cite{Brown2024}, we bound conditional entropy via noncommutative polynomial optimization, implemented through the Navascu\'es--Pironio--Ac\'{\i}n (NPA) hierarchy~\cite{NPA}. This semidefinite programming approach yields near-optimal bounds with negligible gaps to collective attack upper bounds. We also investigate bounds based on min-entropy, which directly quantifies Eve's guessing probability~\cite{Pan-PRL} 
(see SM, Section~\ref{LB_Rate-sec} for a detailed description of the lower bound methods).

Building on these theoretical bounds, we examine how practical imperfections affect protocol performance. In particular, mode mismatch at Charlie's station critically affects performance by reducing interferometric visibility ($V_c$) and degrading heralded entanglement (see SM, Section~\ref{vis-sec}). In addition, local mismatches between the signal modes and the displacement coherent states in the measurement setups at Alice's ($V_a$) and Bob's ($V_b$) local labs can further reduce the observed Bell violation. As shown in SM Section~\ref{local-vis}, their impact is comparable to that of Charlie’s visibility. For simplicity, in our main analysis we assume a common visibility parameter $V:=V_a=V_b=V_c$.

Fig.~\ref{fig:keyrate-main} shows asymptotic key rates versus local detection efficiency $\eta_d$ for various visibility values.
Our numerical optimization reveals critical thresholds: positive key rates require $\eta_d = 91.5\%$ for the 1-photon protocol and $\eta_d=80.2\%$ for the 2-photon one. Importantly, optimization over the parametric gain $g$ enables the 2-photon protocol to approach the ideal state $\ket{00}+r\ket{11}$ with the lowest detection efficiency threshold. Although a smaller $g$ reduces the heralding probability, this effect is fully compensated by increasing the number of rounds in the asymptotic case. Both protocols show similar robustness to imperfect visibility $v<1$, thanks to the fact that a single photon interferes at Charlie’s station, leading to a $\sqrt{V}$-dependence of the joint probabilities. In contrast, in polarization-based protocols ~\cite{Oudot2024, Masini2022}, the interferometric visibility enters linearly, causing a considerably stronger degradation in key-rate performance.

\bigskip

\begin{figure}
    \centering
    \includegraphics[width=\columnwidth]{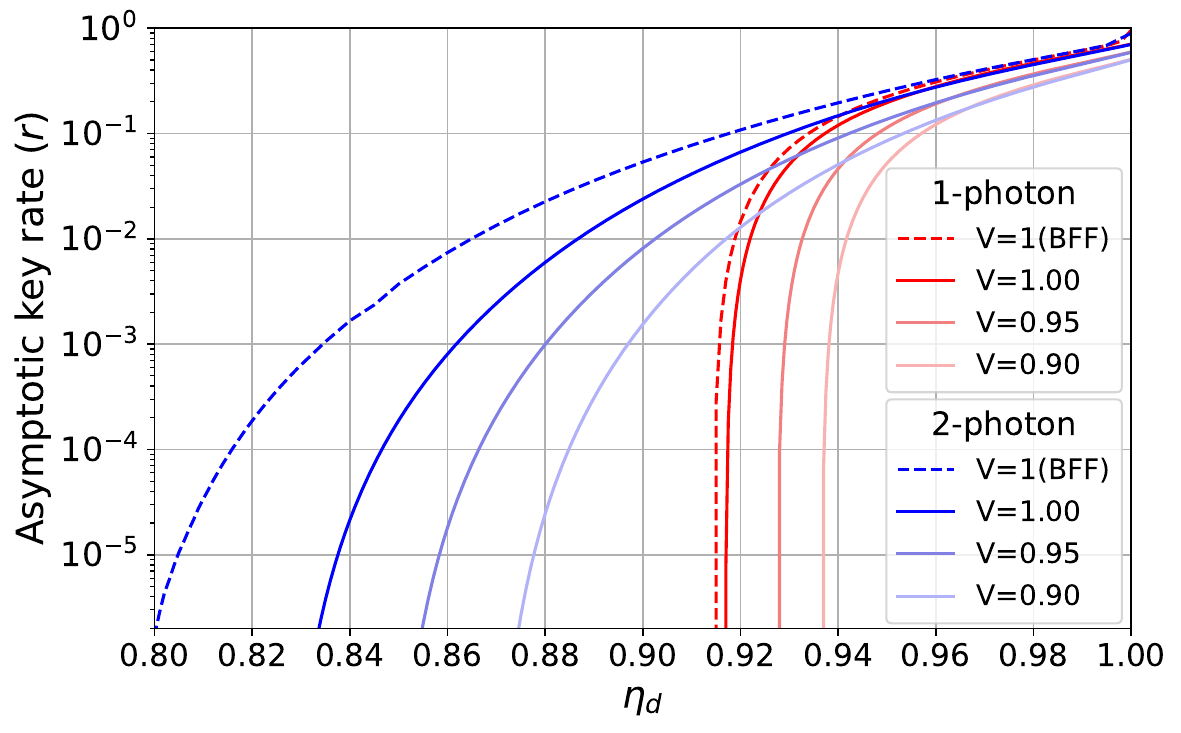}
    \caption{Asymptotic raw key rates $r$ (logarithmic scale) as a function of local detection efficiency $\eta_d$ for visibility $v$ at Charlie's, computed for the 1-photon protocol (red), and the 2-photon protocol (blue), both with optimal preprocessing $q_{\text{opt}}$.
    The solid lines correspond to the key rate computed using the analytical bound in Eq.~\eqref{SKR-CHSH+NP} for $V=1, 0.95, 0.9$, while dashed lines show results from the BFF method (see SM, Section~\ref{LB_Rate-sec}.B) for $V=1$.}
    \label{fig:keyrate-main}
\end{figure}

\noindent
\textit{Finite-key security analysis.---} Practical implementation of our DI-TF QKD protocol requires analysis of finite-size effects to determine maximum achievable distances and realistic key rates while ensuring security against general quantum attacks~\cite{Dupuis2020,Dupuis2019}. We employ EAT to bound conditional entropy $H(A_1\vert E)$ over $N$ rounds~\cite{Dupuis2020,Dupuis2019}.

Here, we implement this analysis through two complementary approaches. First, using Eq.~\eqref{SKR-CHSH+NP}, we construct a linear bound via its tangent. This method offers computational simplicity, requiring only CHSH monitoring. Next, we perform full statistics certification exploiting the complete probability distribution $p(a,b\vert x,y)$, resulting in tighter bounds at increased computational cost. Error correction employs the VHASH hashing algorithm~\cite{Dai2007}, followed by privacy amplification~\cite{NadlingerNature2022}.

The finite-size raw key rate for $N$ rounds is $r=\ell/N$, where $\ell$ denotes the extractable secure key length. We convert it to bits per second (bps) using $R = P_h^\text{}\nu r$, where $\nu = 100$ MHz is the laser source repetition rate and $P_h$ is the heralding probability. SM, Section~\ref{finite-sec} provides details on the numerical approaches used to optimize the finite-key rate.

\begin{figure}
    \centering
    \includegraphics[width=\columnwidth]{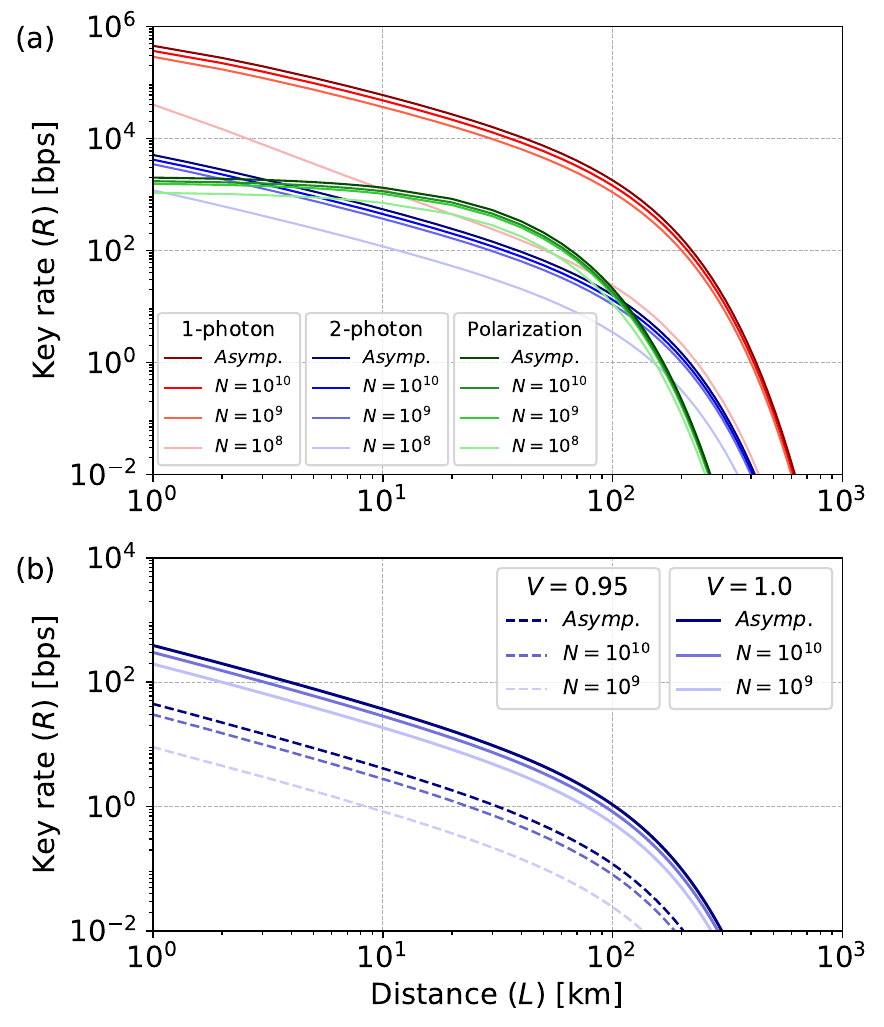}
    \caption{Comparison of finite-size and asymptotic key rates $R$ versus distance $L$ for $N \in \{10^8, 10^9, 10^{10}, \infty\}$ protocol rounds, assuming repetition rate $\nu=100$ MHz, optimal parametric gain $g$, computed for:
    (a) local detection efficiency $\eta_d = 93\%$ for the 1-photon protocol (red), the 2-photon protocol (blue), and the polarization-based protocol from Ref.~\cite{Oudot2024} (green), all for visibility $V=1$, 
    (b) local detection efficiency $\eta_d = 90\%$ for the 2-photon protocol, computed for visibilities $V=0.95$ (dashed lines) and 1 (solid lines).}
    \label{fig:protocol-comparison}
\end{figure}

Fig.~\ref{fig:protocol-comparison}a compares the finite-size key rates as a function of distance across various protocols at an overall efficiency of $\eta_d = \eta_s= 0.93$. Notably, we observe a significant distance breakthrough: the 1-photon protocol maintains key rates above 1~bps beyond 400~km, which is unprecedented for DI QKD protocols. 
Although the asymptotic rates remain positive at all distances, finite-size effects impose practical limitations. With $N \geq 10^{9}$ repetition rounds, the key rates closely approach the asymptotic bound, whereas for $N = 10^8$ we observe their significant decrease.

To explore a more practical scenario, we consider the 2-photon protocol at a reduced detection efficiency of $\eta_d = \eta_s = 0.90$ with visibilities $V = 1$ and $0.95$. As depicted in Fig.~\ref{fig:protocol-comparison}b, setting an optimal $g_a$ and $g_b$  enables this setup to reach approximately 1~bps for $N = 10^{10}$ rounds over a 100~km distance.

\bigskip

\noindent
\textit{Feasibility.---} The protocols studied in this Letter face two main challenges. First, the measurement scheme requires displacement operations combined with single-photon counting. Such measurements have been demonstrated experimentally~\cite{Caspar2020}, with efficiency requirements primarily driven by the need for highly efficient two-mode squeezed-vacuum sources. Second, phase drifts in the heralding scheme can lead to heralding of separable states. Ref.~\cite{Caspar2020} addressed this using phase stabilization in a single-photon interference-based heralding scheme closely related to our protocol.  Solid-state entanglement at metropolitan scales has recently been achieved using single-photon heralding schemes~\cite{Stolk2024,Zhao2024}. Long-distance device-independent quantum key distribution based on high-visibility single-photon interference has been demonstrated very recently~\cite{Lu2026}.

\bigskip

\noindent
\textit{Conclusions.---} We have studied two fully photonic approaches for realizing device-independent quantum key distribution with twin-field scaling. The choice between them involves a trade-off between detector requirements and the key rate achievable at long distances. The 2-photon protocol offers the best efficiency threshold up to 80.2\% (91.5\% for the 1-photon protocol) in the asymptotic regime, achieving state-of-the-art threshold for DI-TF QKD (see SM, Table~\ref{table:det-threshold} for a complete comparison). This is lower than the $81.8\%$ required in polarization-based schemes, although both protocols can reach the Eberhard limit for CHSH violation.
In the realistic finite-size regime, we find that an efficiency of 90\% is required to obtain a high enough key rate. The key rate obtained in the 1-photon protocol is higher by one order of magnitude when the efficiency is higher than 92\%. State-of-the-art experiments can achieve 82.2\% overall efficiency ~\cite{Zhao2024} which makes a proof-of-principle experiment for the 2-photon protocol available with current setups. We believe that our work represents an important step toward the realization of long-distance DI QKD.
\bigskip

\noindent
\textit{Note added:} During the preparation of this manuscript, we became aware of two related works, which study similar protocols, but only in the asymptotic regime. In Ref.~\cite{Alwehaibi2025}, a 2-photon protocol was presented. In contrast, Ref.~\cite{Ishihara2025} analyzes both 1- and 2-photon protocols and includes a comprehensive dark count analysis. Compared to them, our work provides both rigorous finite-size security and visibility analyses.
\bigskip

\noindent
\textit{Acknowledgments.---} M. M. and M. S. were supported by the European Union's Horizon 2020 research and innovation programme under the Marie Skłodowska-Curie project `AppQInfo' No.\ 956071. M. A. and M. S. were supported by the National Science Centre `Sonata Bis' Project No.\ 2019/34/E/ST2/00273, and M. S. by the QuantERA II Programme which has received funding from the European Union's Horizon 2020 research and innovation programme under Grant Agreement No.\ 101017733, Project `PhoMemtor' No.\ 2021/03/Y/ST2/00177. P.M. was supported by the European Union’s Horizon Europe research and innovation programme under grant agreement No.\ 101080086/NeQST. E. O. acknowledges funding by the French national quantum initiative managed by Agence Nationale de la Recherche in the framework of France 2030 with the reference ANR-22-PETQ-0009. NPA optimization was implemented using Python library ncpol2sdpa~\cite{Wittek_2015}, and MOSEK~\cite{mosek} was used as a solver. The Entropy Accumulation analysis was performed using expdiqrng package~\cite{expdiqrng}.

\onecolumngrid
\clearpage

\renewcommand{\thesection}{S\arabic{section}}
\setcounter{section}{0}
\renewcommand{\theequation}{S\arabic{equation}}
\setcounter{equation}{0}
\renewcommand{\thefigure}{S\arabic{figure}}
\setcounter{figure}{0}

\begin{center}
\textbf{\large Supplemental Material:\\\strut Long-range photonic device-independent quantum key distribution\\ using SPDC sources and linear optics}
\end{center}

\begin{quote}
This Supplementary Material provides detailed theoretical and numerical analyses supporting the DI-TF QKD protocol. It presents the mathematical framework for
1-photon and 2-photon entanglement distribution, including derivations of marginal and joint probability distributions under various loss conditions. The document compares three different approaches for calculating secure key rates: the CHSH inequality method, guessing probability with min-entropy, and the Brown--Fawzi--Fawzi (BFF) conditional entropy method. Each method is evaluated with and without post-processing techniques (noisy preprocessing and post-selection). 
The material includes extensive numerical results demonstrating that detection efficiency thresholds range from 80.2\% to 95.9\%, depending on the protocol and numerical approaches. 
The method with the BFF method combined with noisy preprocessing achieves the best performance, with a threshold efficiency of 91.5\% for the 1-photon protocol and 80.2\% for the 2-photon protocol, with a super-low key rate value around the threshold. Additionally, the document examines the impact of visibility between SPDC outputs on protocol performance and provides bounds for certified quantum randomness generation.
\end{quote}

\section{Entanglement distribution}
\label{Entang-sec}
In this section, we derive the density matrices and heralding probabilities for both the one-photon and two-photon protocols, taking into account transmission losses and higher-order photon contributions.

\subsection{The 1-photon protocol}
\label{1phEntang-sec}

In the entanglement distribution protocol shown in Fig.~\ref{fig:setup} (a), Alice and Bob pump their spontaneous-parametric down-conversion (SPDC) crystals with pulsed lasers and generate locally two-mode squeezed vacuum (TMSV) states of the form
\begin{equation}\label{TMSV State}
    \ket{\Psi} = \sum_{n = 0}^{\infty} \sqrt{\lambda_n} \ket{n,n}_{1,2},
\end{equation}
where $\sqrt{\lambda_n}=\frac{\tanh^n g}{\cosh g}$ is the probability amplitude of $n$-photon components, $g$ is the parametric gain of SPDC crystals that is governed by the laser power and phase matching conditions, and indices 1 and 2 denote the signal and idler modes, respectively. While $\{\lambda_n\}$ is a geometrically decreasing sequence with the most probable zero and one-photon events, in this work we consider the full multi-photon form of $\ket{\Psi}$ to avoid the free-sampling assumption in the subsequent Bell test. 

Next, idler modes $a_2$ and $b_2$ are sent to Charlie, interfered on a symmetric beam splitter, measured using photon-number-resolved (PNR) detection, and heralded. In the lossless conditions, as a result of this, a generalized Holland--Burnett state~\cite{Thekkadath2020}, which is near-maximally entangled in the photon number, is shared between Alice's and Bob's signal modes, $a_1$ and $b_1$
\begin{equation}
    \ket{\Psi_\text{out}^{(\sigma,k)}} = 
    \sum_{n=0}^\sigma \mathcal{A}_\sigma(k,n)\,\ket{n,\sigma-n}_{a_1,b_1},
    \label{eq:SIpsiout}
\end{equation}
where $k$ and $\sigma-k$ are the heralded Charlie's readouts, $\mathcal{A}_\sigma(k,n)=i^{k-n}\phi_k(n-\frac{\sigma}{2},\sigma)$ is the probability amplitude, $k=0, \dots, \sigma$, and $\phi_k$ are symmetric Kravchuk functions -- orthonormal discrete polynomials which converge to Hermite--Gauss polynomials for large $\sigma$. For more information on these functions, see the appendices in~Refs.~\cite{PRA} and \cite{QKT}.
If $\sigma=1$, Charlie heralds a single-photon entanglement in one of the following forms:
\begin{align}
     \ket{\Psi_\text{out}^{(1,0)}} =  \tfrac{1}{\sqrt{2}}(\ket{0,1} - i \ket{1,0})_{a_1, b_1},\\
     \ket{\Psi_\text{out}^{(1,1)}} = \tfrac{1}{\sqrt{2}}(\ket{0,1} + i \ket{1,0})_{a_1, b_1}.
\end{align}
Without loss of generality, in this work we focus on $\ket{\Psi_\text{out}^{(1,1)}} \equiv \ket{\Psi_\text{out}^\text{(1ph)}}$.

In a non-ideal scenario, where the idler channels are lossy, the losses can be modeled using a beam splitter with reflectivity $r_t$ and transmissivity $\eta_t$, satisfying $r_t + \eta_t=1$. Denoting $S$ as the total number of photons produced by Alice's and Bob's SPDCs in the idler modes, and $\sigma$ as the total number of photons arriving at Charlie's station after the photon loss, the state from Eq.~\eqref{eq:SIpsiout} converts into a mixed state
\begin{align}\label{rho_out}
    \rho_\text{out}^{(k,\sigma)} = \sum_{S=\sigma}^{\infty}P_{S|\sigma} \rho_S^{(k,\sigma)},
\end{align}
in which $P_{S|\sigma} = \left( r_t \tanh^2 g \right)^{S-\sigma} \left(1 - r_t \tanh^2 g \right)^{\sigma + 2} \binom{S + 1}{\sigma + 1}$, and 
\begin{align} \label{rho_S}
    \rho_S^{(k,\sigma)}
    =\tilde{\mathcal{N}}^{-2}_S\hspace{-2mm}
    \sum_{n,n'=0}^{S} \ket{n,S-n}\bra{n',S-n'} 
    \hspace{-6mm}\sum_{p=\max(0, n-\sigma, n'-\sigma)}^{\min(S-\sigma,n,n')} 
    \scalebox{1.25}{$\sqrt{{ n\choose p}{ n'\choose p} {S-n\choose S-\sigma-p}{S-n'\choose S-\sigma-p}}$} \times \mathcal{A}_\sigma(k,n-p) \mathcal{A}^*_\sigma(k,n'-p),
\end{align}  
where
$\tilde{\mathcal{N}}^{2}_S=\sum_{n=0}^{S}\sum_{p=\max(0, n-1)}^{\min(S-1, n)} {n \choose p} {S-n \choose S-p-1} \lvert\mathcal{A}_\sigma(k, n-p) \rvert^2 = \binom{S+1}{\sigma+1}$. 
Here, the corresponding state for $S=\sigma$ in Eq.~\eqref{rho_out} is given by $\rho_\sigma^{(k,\sigma)} = \ket{\Psi_\text{out}^{(\sigma,k)}}\bra{\Psi_\text{out}^{(\sigma,k)}}$, which remains the dominant term provided that $g \ll 1$.
By changing the summation order and redefining the variables in Eqs.~\eqref{rho_out} and \eqref{rho_S}, the state can be expressed in the following equivalent form:
\begin{align} \label{rho_out_separate}
    \rho_{\text{out}}^{(k,\sigma)}
    = \left(1 - r_t \tanh^2 g \right)^{\sigma + 2} 
    \sum_{p,q=0}^{\sigma} \mathcal{A}_\sigma(k,p) \mathcal{A}^*_\sigma(k,q)
    \sum_{m=0}^{\infty} \big(r_t \tanh^2 g \big)^{m}  
    \sqrt{\binom{m+p}{m}\binom{m+q}{m}} \ket{m+p}\bra{m+q}_{a_1} \nonumber\\
    \otimes \sum_{n=0}^{\infty} \big(r_t \tanh^2 g \big)^{n}  
    \sqrt{\binom{n+\sigma-p}{n}\binom{n+\sigma-q}{n}} \ket{n+\sigma-p}\bra{n+\sigma-q}_{b_1}.
\end{align}
To rewrite the state for the 1-photon protocol by setting $\sigma=k=1$, we have
\begin{align} \label{rho_out_1ph}
    & \rho_{\text{out}}^{(\text{1ph})}
    = \frac{\left(1 - r_t \tanh^2 g \right)^{3}}{2} \times \\
    & \Biggl[\ \sum_{m=0}^{\infty} \big(r_t \tanh^2 g \big)^{m} \ketbra{m}{m}
    \otimes \sum_{n=0}^{\infty} n \big(r_t \tanh^2 g \big)^{n-1} \ketbra{n}{n}
    + \sum_{m=0}^{\infty} m \big(r_t \tanh^2 g \big)^{m-1} \ketbra{m}{m}
    \otimes \sum_{n=0}^{\infty} \big(r_t \tanh^2 g \big)^{n} \ketbra{n}{n}\nonumber\\
    & + \sum_{m=0}^{\infty} \big(r_t \tanh^2 g \big)^{m} \sqrt{m+1} \
    \sum_{n=0}^{\infty} \big(r_t \tanh^2 g \big)^{n} \sqrt{n+1} \
    \biggl(\ket{m}\bra{m+1} \otimes \ket{n+1}\bra{n}
    +\ket{m+1}\bra{m} \otimes \ket{n}\bra{n+1}\biggr) \Biggr]_{a_1,b_1}.
\end{align}
The heralding probability of a 1-photon protocol can be obtained by substituting $\sigma=k=1$ in
\begin{equation}\label{eq:1ph PNR heralding proba}
    P_h(k,\sigma) = \frac{\lambda_\sigma}{\cosh^2 g} . \frac{(1-r_t)^\sigma}{(1-r_t \tanh^2 g)^{\sigma+2}}
    = \frac{(\sqrt{\eta_t} \sinh^2{g})^\sigma}{(1+\sqrt{\eta_t} \sinh^2{g})^{\sigma+2}}.
\end{equation}
which is obtained by assuming a perfect detection at Charlie’s station~\cite{PRA}.
Using $r_t = 1-10^{-\alpha_\text{att}(\frac{L}{2})/10} = 1-\sqrt{\eta_t}$, where $\alpha_\text{att}$ is the attenuation coefficient of a medium through which the photons travel.  For standard telecom fibers, $\alpha_\text{att} = 0.2 \; \mathrm{km}^{-1}$, and $L$ is the total distance between Alice and Bob in kilometers. 
Assuming that Charlie is located approximately at the midpoint between Alice and Bob, we replaced $L_{AC}$ with $L/2$. 
Thus, the total probability of the generation of the 1-photon entanglement per pulse is

\begin{equation}
    P_h^{(\text{1ph})} =
    \frac{10^{L/50} \sinh ^2 g}{\left[10^{L/100}+\sinh^2 g\right]^3} \eta_c \, 
    {\approx} \, 10^{-L/100} \sinh^2{g} \ \eta_c\,
    \propto \, O(\sqrt{\eta_t}),
\end{equation}
where $\eta_c$ accounts for the detection efficiency at Charlie's. 
The approximation holds for long distances, as $\sinh^2 g \ll 10^{L/100}$ under the fact that $g \ll 1$, and the protocol scaling is consistent with the twin-field (TF) QKD performance.

Note that the probability ratio of multi-photon to single-photon arrivals at Charlie can be expressed as
\begin{equation}
    \scalebox{1.1}{$\frac{1-P_h(0,0)-P_h(0,1)-P_h(1,1)}{P_h(0,1)+P_h(1,1)}
    =\frac{(1+\sqrt{\eta_t}\sinh^2{g})^{3}-(1+\sqrt{\eta_t}\sinh^2{g})-(2\sqrt{\eta_t}\sinh^2{g})}{2\sqrt{\eta_t}\sinh^2{g}} =\frac{1}{2} $}
    \scalebox{1}{$\sqrt{\eta_t}\sinh^2{g} \ (3+\sqrt{\eta_t}\sinh^2{g})$}.
\end{equation}
Since this ratio scales as $o(10^{-L/100} g^2 )$, when $g \ll 1$ or the distance is large ($L \gg 1$), replacing the PNR detectors with the on/off detectors leaves $\rho^{(\text{1ph})}_\text{out}$ unchanged.

\subsection{The 2-photon protocol}
\label{2phEntang-sec}

The 2-photon protocol, described in this work extends the 1-photon scheme by employing an asymmetric source configuration where Alice prepares a TMSV state
{$\big(\ket{\psi_a} = \sum_{m = 0}^{\infty} \sqrt{\lambda_m^{(a)}}\ \ket{m,m}_{a_1,a_2}\big)$}, 
while Bob ideally generates a heralded single-photon path-entangled (SPPE) state, 
$\big(\ket{\psi_b} = \sqrt{t_s}\ket{0,1}_{b_1,b_2}+ e^{i\phi_s} \,\sqrt{1-t_s}\,\ket{1,0}_{b_1,b_2}\big)$. 
Upon successful heralding at Charlie's station through single-photon detection, the protocol produces a tunable photon-number path-entangled state between Alice and Bob
\begin{equation}
\ket{\Psi_\text{out}^{(\text{2ph})}} \propto 
\sqrt{\lambda_0^{(a)} t_s t_c} |0, 0\rangle_{a_1,b_1} +  e^{i\phi} \, \sqrt{\lambda_1^{(a)} r_s r_c} |1, 1\rangle_{a_1,b_1},
\end{equation}
where $t_{s(c)}$ and $r_{s(c)}$ are Bob's (Charlie's) beam splitter transmissivity and reflectivity, respectively (with  $t_{s(c)}+r_{s(c)}=1$), and $\phi=\phi_s+\phi_c$ are the combined phases arising from Bob's and Charlie's beam splitters. This state can be continuously tuned from a product state to a maximally entangled state by adjusting the source parameters.

In the non-ideal case, where Bob employs an imperfect `photons'/`no-photons' (on/off) detector of efficiency $\eta_s$ on beam $b_0$ to herald a single-photon state, the resulting density matrix in mode $b_1$ is
\begin{equation}
    \rho^{(b)}_{0} 
    = \mathcal{N}_b^{-2}\mathrm{Tr}_{b_0}\Big[ 
    \big(E_1^{\eta_s} \otimes \mathbb{I}\big)_{b_0,b_2} \ \ketbra{\Psi}{\Psi} \Big]
\end{equation}
where $\ket{\Psi}$ is defined in Eq.~\eqref{TMSV State}, and $E_1^{\eta_s}=\mathbb{I}-\sum_{k=0}^\infty (1-\eta_s)^k\ketbra{k}{k}$ represents the effective POVM element for click detection. Thus, the heralded state in mode $b_2$ simplifies to
\begin{equation}
    \rho^{(b)}_{0} 
    = \mathcal{N}_b^{-2} \sum\nolimits_{n=0}^\infty \lambda^{(b)}_n \big (1-(1-\eta_s)^n \big)\ketbra{n}{n}
\end{equation}
with normalization factor $\mathcal{N}_b^2= 1-(1+\eta_s \sinh^2{g_b})^{-1} = P_s$, representing the probability that Bob’s detector clicks and the imperfect SPPE state is generated.
Then, Bob interferes this state with a vacuum in mode $b_1$ at a beam splitter, resulting in a path-entangled state across modes $b_1$ and $b_2$ given by
\begin{align}
    \rho^{(b)}_\text{in} 
    & = U^{(b_1,b_2)}_\text{BS} \ (\ketbra{0}{0} \otimes \rho^{(b)}_{0})_{b_1,b_2} {U^{(b_1,b_2)}_\text{BS}}^\dagger  \nonumber\\
    & = \mathcal{N}_b^{-2} \sum_{n=0}^\infty \lambda^{(b)}_n \big (1-(1-\eta_s)^n \big)
    \sum_{l,l'=0}^{n} \sqrt{\binom{n}{l}\binom{n}{l'} t_s^{l+l'} r_s^{2n-l-l'}} e^{i\phi_s(l'-l)}\ketbra{n-l,l}{n-l',l'}_{b_1,b_2} 
\end{align}

Moreover, idler channels experience symmetric loss, i.e.\ $r_t^{(a_2)}=r_t^{(b_2)}$, which can be modeled using a beam splitter $U_\text{BS}^{a_2,a_3 (b_2,b_3)}$, through which each idler mode $a_2(b_2)$ interferes with a vacuum mode $\ket{0}_{a_3(b_3)}$.  
Then, after the beams interfere at Charlie’s station and a single-photon detection occurs, the resulting shared mixed state can be calculated as
\begin{align}
    \rho_{\text{out}}^{(\text{2ph})} = {\mathcal{N}}^{-2} 
    \big(\sqrt{t_c}\bra{01}+e^{-i\phi_c}\sqrt{r_c}\bra{10} \big)_{a_2,b_2}
    &\Bigg[\mathrm{Tr}_{a_3}\Big( U_\text{BS}^{(a_2,a_3)}
    (\ketbra{\psi_a}{\psi_a}_{a_1,a_2} \otimes\ketbra{0}{0}_{a_3}) 
    {U_\text{BS}^{(a_2,a_3)}}^\dagger\Big)\\
    & \otimes \mathrm{Tr}_{b_3}\Big( U_\text{BS}^{(b_2,b_3)}
    (\rho_{in}^{(b)} \otimes\ketbra{0}{0}_{b_3}) 
    {U_\text{BS}^{(b_2,b_3)}}^\dagger \Big) \Bigg]
    \big(\sqrt{t_c}\ket{01}+e^{i\phi_c}\sqrt{r_c}\ket{10} \big)_{a_2,b_2}. \nonumber
\end{align}
Using 
$\mathrm{Tr}_{x_3} \big(U_\text{BS}\ketbra{m,0}{m',0}_{x_2,x_3} U_\text{BS}^\dagger \big) 
= \sum_{p=0}^{\min(m,m')} \sqrt{\binom{m}{p} \binom{m'}{p}\ (1-r_t)^{m+m'-2p} \ r_t^{2p}} \,\ketbra{m-p}{m'-p}_{x_2}$, we have:
\begin{align}
    \rho_{\text{out}}^{(\text{2ph})} & = {(\mathcal{N N}_b})^{-2} 
    \big(\sqrt{t_c}\bra{01} \ +\ e^{-i\phi_c}\sqrt{r_c}\bra{10} \big)_{a_2,b_2}
    \Bigg[\sum_{m,m'=0}^{\infty} \sqrt{\lambda^{(a)}_m \lambda^{(a)}_{m'}} 
    \sum_{p=0}^{\min(m,m')} \sqrt{\binom{m}{p} \binom{m'}{p}(1-r_t)^{m+m'-2p} \ r_t^{2p}}
    \nonumber\\
    & \times \ketbra{m,m-p}{m',m'-p}_{a_1,a_2} 
    \otimes \sum_{n=0}^\infty \lambda^{(b)}_n \big (1-(1-\eta_s)^n \big)
    \sum_{l,l'=0}^{n} \sqrt{\binom{n}{l}\binom{n}{l'} t_s^{l+l'} r_s^{2n-l-l'}} \ e^{i\phi_s(l'-l)} \\
    & \times \sum_{q=0}^{\min(l,l')} \sqrt{\binom{l}{q} \binom{l'}{q}\ (1-r_t)^{l+l'-2q} \ r_t^{2q}} \ \ketbra{n-l,l-q}{n-l',l'-q}_{b_1,b_2} \Bigg]
    \big(\sqrt{t_c}\ket{01}+e^{i\phi_c}\sqrt{r_c}\ket{10} \big)_{a_2,b_2}. \nonumber
\end{align}

\noindent By performing a straightforward calculation, the state simplifies to
\begin{align}\label{rho_out_2ph}
    & \hspace{-2mm}\rho_{\text{out}}^{(\text{2ph})} = {(\mathcal{N N}_b)}^{-2} (1-r_t)
    \sum\nolimits_{n=0}^\infty \lambda^{(b)}_n \big (1-(1-\eta_s)^n \big) \\
    & \hspace{-1mm} \scalebox{0.91}{$\times \Bigg[ t_c\sum\limits_{m=0}^\infty \lambda^{(a)}_m r_t^m \ketbra{m}{m}_{a_1}
    \otimes \sum\limits_{l=0}^{n} \binom{n}{l} t_s^{l} r_s^{n-l} \times l r_t^{l-1} \ketbra{n-l}{n-l}_{b_1}
    + r_c \sum\limits_{m=0}^\infty \lambda^{(a)}_m \, m r_t^{m-1} \ketbra{m}{m}_{a_1}
    \otimes \sum\limits_{l=0}^{n} \binom{n}{l} t_s^{l} r_s^{n-l} \times r_t^{l} \, \ketbra{n-l}{n-l}_{b_1}$} \nonumber\\
    & \hspace{-1mm} \scalebox{0.93}{$ +\ \sqrt{t_c r_c} e^{i(\phi_c-\phi_s)}
    \sum\limits_{m=0}^\infty \sqrt{\lambda^{(a)}_m \lambda^{(a)}_{m+1} (m+1)} \ r_t^m \ketbra{m}{m+1}_{a_1}
    \otimes \sum\limits_{l'=0}^{n-1} \sqrt{\binom{n}{l'+1}\binom{n}{l'} t_s^{2l'+1} \ r_s^{2n-2l'-1} \times (l'+1) \ r_t^{2l'}} \ketbra{n-l'-1}{n-l'}_{b_1} $}\nonumber\\
    & \hspace{-1mm} \scalebox{0.93}{$ +\ \sqrt{r_c t_c} e^{i(\phi_s-\phi_c)}
    \sum\limits_{m'=0}^\infty \sqrt{\lambda^{(a)}_{m'+1} \lambda^{(a)}_{m'} (m'+1)} \ r_t^{m'} \ketbra{m'+1}{m'}_{a_1}
    \otimes \sum\limits_{l=0}^{n-1} \sqrt{\binom{n}{l}\binom{n}{l+1} t_s^{2l+1} \ r_s^{2n-2l-1} \times (l+1) \ r_t^{2l}} \ketbra{n-l}{n-l-1}_{b_1} \Bigg].$}
    \nonumber
\end{align}
The normalization factor ${(\mathcal{N N}_b)}^2$ can be calculated by taking trace over $\rho_{\text{out}}^{(\text{2ph})}$:
\begin{align}
    \hspace{-2mm} {(\mathcal{N N}_b)}^2 & = 
    (1-r_t) \sum_{n=0}^\infty \lambda^{(b)}_n \big (1-(1-\eta_s)^n \big)
    \Bigg[ t_c\sum\limits_{m=0}^\infty \lambda^{(a)}_m r_t^m
    \sum\limits_{l=0}^{n} \binom{n}{l} t_s^{l} r_s^{n-l} .\,l r_t^{l-1}
    + r_c \sum\limits_{m=0}^\infty \lambda^{(a)}_m \, m r_t^{m-1}
    \sum\limits_{l=0}^{n} \binom{n}{l} t_s^{l} r_s^{n-l} .\, r_t^{l}\Bigg] \nonumber\\
    & = (1-r_t) \sum_{n=0}^\infty \lambda^{(b)}_n \big (1-(1-\eta_s)^n \big)
    \Bigg[ \frac{t_c \cosh^{-2}{g_a}}{ 1-r_t \tanh^2{g_a}} 
    .\frac{\partial (r_s + t_s r_t)^n}{\partial r_t}
    + \frac{r_c}{\cosh^2{g_a}} \frac{\partial (1-r_t \tanh^2{g_a})^{-1}}{\partial r_t}
    . (r_s + t_s r_t)^n \Bigg] \nonumber\\
    & = (1-r_t)
    \Bigg[ \frac{t_c \cosh^{-2}{g_b}}{1+(1-r_t) \sinh^2{g_a}} 
    .\frac{\partial}{\partial r_t} 
    + \frac{r_c \sinh^2{g_a} \cosh^{-2}{g_b}}{(1+(1-r_t) \sinh^2{g_a})^2} \Bigg]
    \Big(\sum_{n=0}^\infty \big (1-(1-\eta_s)^n \big)\big((r_s + t_s r_t) \tanh^2{g_b} \big)^n\Big) \nonumber \\
    & \! = (1-r_t) \scalebox{1.2}{$
    \! \Bigg[ \frac{t_c}{1+\sqrt{\eta_t} \sinh^2{g_a}} 
    .\frac{\partial}{\partial r_t} 
    + \frac{r_c \sinh^2{g_a}}{(1+\sqrt{\eta_t} \sinh^2{g_a})^2} \Bigg]
    \! \Big(\frac{\cosh^{-2}{g_b}}{1-(r_s + t_s r_t) \tanh^2{g_b}} - \frac{\cosh^{-2}{g_b}}{1-(1-\eta_s)(r_s + t_s r_t) \tanh^2{g_b}}\Big)$}
\end{align}
where we substituted $\lambda_n^{a(b)} = \tanh^{2n}{g_{a(b)}}/\cosh^2{g_{a(b)}}$ and used $\sum_{n=0}^\infty nx^{n-1} = \frac{\partial}{\partial x} (\sum_{n=0}^\infty x^n) = (1-x)^{-2}$ for $x<1$. 
Using $1-r_t=10^{-\alpha_{att}(\frac{L}{2})/10} = \sqrt{\eta_t}$ and replacing $\mathcal{N}_b^2= \eta_s \sinh^2{g_b}/(1+\eta_s \sinh^2{g_b})$, we then obtain
\begin{align}
    \mathcal{N}^2 & = \frac{\sqrt{\eta_t} (1+\eta_s\sinh^2{g_b})}
    {(1+\sqrt{\eta_t}\sinh^2{g_a}) (1+\sqrt{\eta_t} \,t_s\sinh^2{g_b}) \big(1+(\eta_s + (1-\eta_s) \sqrt{\eta_t} \,t_s)\sinh^2{g_b} \big)} \nonumber\\
    & \times \Bigg[\frac{t_c \Big[(1+\sinh^2{g_b})^2-(1-\eta_s)(1-\sqrt{\eta_t} \, t_s)^2 \sinh^4{g_b}\Big] t_s }
    {(1+\sqrt{\eta_t}\,t_s\sinh^2{g_b}) \big(1+(\eta_s + (1-\eta_s) \sqrt{\eta_t} \,t_s)\sinh^2{g_b} \big)} 
    + \frac{r_c (1-\sqrt{\eta_t}\,t_s) \sinh^2{g_a}}{1+\sqrt{\eta_t}\sinh^2{g_a}}\Bigg].
\end{align}
which, for long distances where $\eta_t \ll 1$, is proportional to $o(\sqrt{\eta_t})$. 
In the limit $g_b\ll 1$, where Bob's initial state approaches the SPPE state, the shared state and its normalization factor reduce to 
\begin{align}\label{rho_out_2ph_SPPE}
    \tilde{\rho}_{\text{out}}^{(\text{2ph})} = & \, \tilde{\mathcal{N}}^{-2} \sqrt{\eta_t} 
    \Big[\sum_{n=0}^\infty \lambda^{(a)}_n r_t^n (t_st_c + n t_s r_c) \ketbra{n,0}{n,0}
    + \frac{r_s r_c}{r_t} \sum_{n=0}^\infty n\lambda^{(a)}_n r_t^n\ketbra{n,1}{n,1} \nonumber\\
    &\qquad{}+ \sqrt{t_s r_s t_c r_c} \, \tanh{g_a} \sum_{n=0}^\infty \sqrt{n+1} 
    \Big(e^{i(\phi_c-\phi_s)} \ketbra{n,0}{n+1,1} 
    + e^{i(\phi_s-\phi_c)} \ketbra{n+1,1}{n,0}\Big) \Big]_{a_1,b_1},
\end{align}
\begin{equation}
    \tilde{\mathcal{N}}^2 = 
    \sqrt{\eta_t} \sum_{n=0}^\infty \lambda_n r_t^n \left(t_st_c + n t_s r_c+ \frac{n r_s r_c}{r_t}\right)
    = \sqrt{\eta_t} \Big(\frac{t_c t_s}{1+\sqrt{\eta_t}\sinh^2{g_a}} 
    + \frac {r_c (1- t_s \sqrt{\eta_t}) \sinh^2{g_a}}{(1+\sqrt{\eta_t} \sinh^2{g_a})^2} \Big),
\end{equation}
The probability of detecting a single photon at Charlie's station is given by the product of Bob's single-photon generation probability $P_s$, the photon arrival probability at Charlie's ${\tilde{\mathcal{N}}}^2$, and his detector efficiency $\eta_c$. 
Assuming a symmetric beam splitter at Charlie's (i.e., $r_c=t_c=\frac{1}{2}$), the heralding efficiency is given by:
\begin{equation}
    P_h^{(\text{2ph})} =
    \tilde{\mathcal{N}}^2 P_s \eta_c=
    \frac{\sqrt{\eta_t} (t_s + \sinh ^2 g_a )}{2\big(1+\sqrt{\eta_t}\sinh^2{g_a}\big)^2} P_s \eta_c \,  
    \propto \, O(\sqrt{\eta_t}).
   \label{eq:2ph heralding proba} 
\end{equation}
For the probability ratio of multi-photon to single-photon arrivals at Charlie, noting that
\begin{equation}
    P_h(0,0)=\frac{(1+\eta_s\sinh^2{g_b}) (1-\sqrt{\eta_t}\, t_s)}
    {(1+\sqrt{\eta_t}\sinh^2{g_a}) (1+\sqrt{\eta_t} \,t_s\sinh^2{g_b}) \big(1+(\eta_s + (1-\eta_s) \sqrt{\eta_t} \,t_s)\sinh^2{g_b} \big)}
    \overset{g_b\ll1}{\approx} \frac{1-\sqrt{\eta_t}\, t_s}{1+\sqrt{\eta_t}\sinh^2{g_a}},
\end{equation}
in the limit $g_b\ll1$, we can write
\begin{equation}
    \scalebox{1.15}{$\frac{1-P_h(0,0)-P_h(0,1)-P_h(1,1)}{P_h(0,1)+P_h(1,1)}
    =\frac{(1+\sqrt{\eta_t}\sinh^2{g_a})^{2}-(1+\sqrt{\eta_t}\sinh^2{g_a})(1-\sqrt{\eta_t}\, t_s) -\sqrt{\eta_t}(t_s+\sinh^2{g_a})}{\sqrt{\eta_t}(t_s+\sinh^2{g_a})}$}
    = \sqrt{\eta_t}\sinh^2{g_a}.
\end{equation}
Since this ratio scales as $o(10^{-L/100} g_a^2 )$, for $g_{a(b)} \ll 1$ or $L \gg 1$, substituting the PNR detectors with the `photons'/`no-photons' (on/off) detectors does not affect $\rho^{(\text{2ph})}_\text{out}$.

\section{Measurements}
\label{POVMs-sec}

 Alice and Bob implement displacement operations $D(\delta) =  e^{\delta a^{\dagger} - \delta^{*} a}$, followed by `photons'/`no-photons' (on/off) detection with efficiency $\eta_d$. 
The parties perform their local measurements without quantum memories in the delayed-choice scheme. To apply different measurement settings, each party interferes their signal mode with a coherent pulse $\ket{\alpha}$ or $\ket{\beta}$ on variable beam splitters characterized by reflectivity-to-transmissivity ratios $r_a:t_a$ and $r_b:t_b$, respectively.
This procedure effectively implements the displacements 
$\delta_a = -i\alpha\sqrt{r_a}$ and $\delta_b = -i\beta\sqrt{r_b}$, yielding the final state
$D(\delta_\alpha)\otimes D(\delta_\beta)\ket{\Psi_\text{out}}$ which is then measured using local detectors.
To account for the imperfection of detectors, we model the effective no-click POVM as a beam splitter followed by a perfect detector, whose inefficiency is represented by losses in the beam splitter. Thus, we can formulate it as
\begin{equation}
    E_{0}^{\eta_{d}} = \sum_{m=0}^{\infty} (1-\eta_d)^{m} \ketbra{m}{m} = (1-\eta_d)^{a^\dagger a},
\end{equation}
where term $\ket{m}\bra{m}$ corresponds to receiving $m$ photons, and $(1-\eta_d)^m$ corresponds to losing all of them. Thus, the effective local measurement for the output 0 can be written as
$\mathcal{M}_0 = D(\delta) E_{0}^{\eta_d} D^{\dagger}(\delta)$.
 To proceed, we need to rewrite and simplify $\mathcal{M}_0$ in the photon-number (Fock) basis:
\begin{equation}\label{Effective Measurements}
    \bra{n} \mathcal{M}_0 \ket{n'}
    = \bra{n}  D(\delta) E_{0}^{\eta} D^{\dagger}(\delta) \ket{n'} 
    = \sum_{m=0}^{\infty} (1-\eta_d)^m \bra{n}  D(\delta) \ketbra{m}{m} D^{\dagger}(\delta) \ket{n'}.
\end{equation}
Continuing the calculation requires representing the displacement operator $D(\delta)$ in the photon number basis as well:
\begin{align} \label{Displacement}
    \bra{n} D(\delta) \ket{m} & 
    = \bra{n} D(\delta)\frac{(a^{\dagger})^{m}}{\sqrt{m!}} \ket{0} 
    = \bra{n} \frac{(a^{\dagger}-\delta^{*})^{m}}{\sqrt{m!}} D(\delta)\ket{0} 
    = \sum_{k=0}^{\min(m,n)} { m \choose k} \frac{(-\delta^{*})^{m-k}}{\sqrt{m!}} \bra{n} (a^{\dagger})^{k} \ket{\delta} \nonumber \\
    & = \hspace{-2mm}
    \sum_{k=0}^{\min(m,n)} { m \choose k} \frac{(-\delta^{*})^{m-k}}{\sqrt{m!}} \sqrt{\frac{n!}{(n-k)!}} e^{-\frac{|\delta|^2}{2}} \frac{\delta^{n-k}}{\sqrt{(n-k)!}} 
    = e^{-\frac{|\delta|^2}{2}} \sum_{k=0}^{\min(m,n)}
     \frac{\sqrt{m!n!} \hspace{2mm} \delta^{n-k} (-\delta^{*})^{m-k}}{k!(m-k)!(n-k)!}. 
\end{align}
Substituting Eq.~\eqref{Displacement} to Eq.~\eqref{Effective Measurements} we obtain
\begin{align} \label{Effective Measurements summations}
    \bra{n} \mathcal{M}_0 \ket{n'} &
    = e^{-|\delta|^2}\sum_{m=0}^{\infty} (1-\eta_d)^m
    \sum_{k=0}^{\min(m,n)} \frac{\sqrt{m!n!} \hspace{2mm} \delta^{n-k} (-\delta^{*})^{m-k}}{k!(m-k)!(n-k)!}
    \sum_{k'=0}^{\min(m,n')} \frac{\sqrt{m!n'!} \hspace{2mm} (\delta^*)^{n'-k'} (-\delta)^{m-k'}}{k'!(m-k')!(n'-k')!} \nonumber\\ &
    = e^{-|\delta|^2} \sqrt{n! n'!} \hspace{2mm}\delta^n (\delta^*)^{n'} \hspace{1mm}
    \sum_{k=0}^{n} \frac{(-|\delta|^2)^{-k}}{k!(n-k)!} \hspace{1mm}
    \sum_{k'=0}^{n'} \frac{(-|\delta|^2)^{-k'}}{k'!(n'-k')!} 
    \sum_{m=\max(k,k')}^{\infty} \frac{m! \big((1-\eta_d)|\delta|^2\big)^m }{(m-k)!(m-k')!}.
\end{align}
Applying the following Lemma~\ref{lemma:equations}.(a) to \eqref{Effective Measurements summations}, the expression for $\mathcal{M}^0$ simplifies to
\begin{align} \label{Local Measurement Matrix}
    \bra{n} \mathcal{M}_0 \ket{n'} &
    = e^{-\eta_d |\delta|^2} \sqrt{n! n'!} \hspace{2mm}\delta^n (\delta^*)^{n'} \hspace{1mm}
    \sum_{k=0}^{n} \frac{(\eta_d-1)^{k}}{k!(n-k)!} \hspace{1mm}
    \sum_{k'=0}^{n'} \frac{(\eta_d-1)^{k'}}{k'!(n'-k')!}
    \sum_{i=0}^{\min(k,k')} \frac{i!}{\big((1-\eta_d)|\delta|^2\big)^i} {k \choose i}{k'\choose i} 
    \nonumber \\ &
    = e^{-\eta_d |\delta|^2} \sqrt{n! n'!} \hspace{2mm}\delta^n (\delta^*)^{n'} \hspace{1mm}
    \sum_{i=0}^{\min(n,n')} \frac{(1-\eta_d)^{i}}{i! |\delta|^{2i}}
    \sum_{k=i}^{n} \frac{(\eta_d-1)^{k-i}}{(n-k)!(k-i)!}
    \sum_{k'=i}^{n'} \frac{(\eta_d-1)^{k'-i}}{(n'-k')!(k'-i)!} \nonumber \\ &
    = e^{-\eta_d |\delta|^2} \sqrt{n! n'!} \hspace{2mm} \eta_d^{n+n'}\delta^n (\delta^*)^{n'} \hspace{1mm}
    \sum_{i=0}^{\min(n,n')} \frac{(1-\eta_d)^{i}}{i!(n-i)!(n'-i)! {(\eta_d |\delta|)}^{2i}},
\end{align}
where in the second line, we changed the order of summations, and in the last line, we used binomial expansion of\break $(1+(\eta_d-1))^{n-i}$ and $(1+(\eta_d-1))^{n'-i}$.

\medskip\noindent
\begin{lemma}\label{lemma:equations}
For a real number $x\in[0,1)$ and non-negative integers m, k, and k', the following equations hold
\begin{enumerate}[label=\textnormal{(\alph*)}]
    \item 
    \begin{align}
        \sum_{n=\max(k,k')}^{\infty} \frac{n! \ x^n }{(n-k)!(n-k')!}
        & = e^{x} \sum_{i=0}^{\min(k,k')} \frac{k! \, k'! \ x^{k+k'-i}}{i! (k-i)! (k'-i)!} \nonumber\\
        \xRightarrow{k=k'} 
        \ \sum_{n=k}^{\infty} \frac{n! \ x^n }{(n-k)!^2}
        & = e^{x} x^k \sum_{i=0}^{k} \frac{k!}{i!}\binom{k}{i} x^{i},\nonumber
    \end{align}
    
    \item 
    \begin{displaymath}
        \sum_{n=k}^{\infty} \frac{n.n! \ x^n}{(n-k)!^2}
        = e^x x^k \sum_{i=0}^k \frac{k!}{i!}\Big[\binom{k+1}{i+1} x^{i+1} + k \binom{k}{i} x^i \Big],
    \end{displaymath}
    
    \item 
    \begin{displaymath}
        \sum_{n=k}^{\infty} {n \choose k} x^n  
        = \frac{x^k}{(1-x)^{k+1}},
    \end{displaymath}
    
    \item
    \begin{displaymath}
        \sum_{n=k}^{\infty} n {n \choose k} x^{n-1}  
        = \frac{x^k + k x^{k-1}}{(1-x)^{k+2}}.
    \end{displaymath}
\end{enumerate}
\end{lemma}

\begin{proof}
\begin{enumerate}[label=\textnormal{(\alph*)}]
   \item By defining $m:= \min(k,k')$ and $M:=\max(k,k')$, the LHS could be written as
    \begin{align}
        \text{LHS} & = x^m \sum_{i=0}^\infty \frac{x^{i+M-m} (i+M)!}{i!(i+M-m)!} 
        = x^m \frac{\partial^m}{\partial x^m} \Big(\sum_{i=0}^\infty \frac{x^{i+M}}{i!} \Big)
        = x^m \frac{\partial^m}{\partial x^m} (x^M e^x) \nonumber\\
        & = x^m \sum_{i=0}^m \binom{m}{i} \frac{\partial^i (x^M)}{\partial x^i} \frac{\partial^{m-i} (e^x)}{\partial x^{m-i}}
        = e^x x^m \sum_{i=0}^m \binom{m}{i} \frac{M! \ x^{M-i}}{(M-i)!} = RHS\ . \nonumber
    \end{align}
            
    \item We can divide the expression in  LHS into the 2 terms:
    \begin{displaymath}
        \sum_{n=k}^{\infty} \frac{n.n! \ x^n}{(n-k)!^2} = \sum_{n=k}^\infty \frac{n! \ x^n}{(n-k)! (n-k-1)!}
        + k \sum_{n=k}^{\infty} \frac{n! \ x^n}{(n-k)!^2}.
    \end{displaymath}
    Applying part (a) with $k'=k+1$ for the first term, and $k'=k$ for the second, we have:
    \begin{displaymath}
        \text{LHS} = e^x x^k \sum_{i=0}^{k} \Big[\frac{k!(k+1)!\ x^{k+1-i}}{i!(k-i)!(k+1-i!)}
        + \frac{k!^2\ x^{k-i}}{i!(k-i)!^2} \Big].
    \end{displaymath}
    The proof is complete upon applying the substitution $i\rightarrow k-i$.
    
    \item It follows by taking the $k-$th derivative of the identity $\sum_{n=0}^\infty x^n=(1-x)^{-1}$, which is valid for $|x|<1$:
    \begin{displaymath}
        \sum_{n=k}^{\infty} {n \choose k} x^{n} = \frac{x^k}{k!}\sum_{n=k}^{\infty} \frac{n! \ x^{n-k}}{(n-k)!}
        = \frac{x^k}{k!} \frac{\partial^k}{\partial x^k} \big( \sum_{n=0}^\infty x^n \big) 
        = \frac{x^k}{k!} \frac{\partial^k (1-x)^{-1}}{\partial x^k} 
        = \frac{x^k}{(1-x)^{k+1}}.
    \end{displaymath}
    
    \item
    \begin{displaymath}
        \sum_{n=k}^{\infty} n {n \choose k} x^{n-1}  
        = \frac{\partial}{\partial x} 
        \Bigl[\sum_{n=k}^{\infty} {n \choose k} x^n \Bigr] 
        = \frac{\partial}{\partial x} 
        \Bigl[\frac{x^k}{(1-x)^{k+1}}\Bigr]
        = \frac{x^{k-1} (x+k)}{(1-x)^{k+2}}.
    \end{displaymath}
    \end{enumerate}
\end{proof}

\section{Derivation of marginals and joint probability distributions}
\label{Distributions-sec}
This section presents the derivation of an analytical formula for marginals and joint probabilities of detecting zero photons for the 1-photon and 2-photon protocols under non-ideal assumptions: lossy transmission channels and imperfect local detection.

\subsection{The 1-photon protocol} \label{App:Imperfect-lossy-1ph}

\noindent
\textbf{I) Marginals:}
We first need to write the reduced density matrix of the subsystem $A$ (or $B$) by tracing out the complementary subsystem ${B}$ (or $A$) from the $\rho_\text{out}^{(k,\sigma)}$, as
\begin{align}
    \rho_\text{out}^{(a)}
    & = \mathrm{Tr}_{B} \big( \rho_\text{out}^{(k,\sigma)}\big) =
    \frac{(1-r_t \tanh^2{g})^{\sigma+2}}{(r_t \tanh^2{g})^\sigma}
    \sum_{S=\sigma}^{\infty} (r_t \tanh^2{g})^S
    \sum_{n=0}^{S} \ket{n}\bra{n}
    \sum_{p=\max(0,n-\sigma)}^{\min(S-\sigma,n)} {n \choose p}{S-n \choose S-\sigma-p} |\mathcal{A}_\sigma(k,n-p)|^2.
\end{align}
By changing the order of summations, shifting the indices $p\rightarrow n-p$ and $S \rightarrow S+n$, and defining $C:= r_t \tanh^2{g}$, we obtain:
\begin{align}
    \rho_\text{out}^{(a)}
    &  = \frac{(1-C)^{\sigma+2}}{C^\sigma}
    \sum_{p=0}^{\sigma} {n \choose p}|\mathcal{A}_\sigma(k,p)|^2 \hspace{1mm}
    \sum_{n=p}^{\infty} C^n \ket{n}\bra{n}
    \sum_{S=\sigma-p}^{\infty} C^S {S \choose \sigma-p} \nonumber\\
    &  = (1-C)
    \sum_{p=0}^{\sigma} \left(\frac{1-C}{C}\right)^p |\mathcal{A}_\sigma(k,p)|^2 \hspace{1mm}
    \sum_{n=p}^{\infty} {n \choose p} C^n \ket{n}\bra{n}.
\end{align}
In the second line we used the identity $\sum_{S=k}^{\infty} x^S {S \choose k} = x^k/(1-x)^{k+1}$ from Lemma~\ref{lemma:equations}.(c). 
Now, by substituting reduced density matrix $\rho_\text{out}^{(x)}$ into $\mathrm{Tr} \big(\rho_\text{out}^{(x)}\, \mathcal{M}_0 \big)$, where $x$ corresponds either to subsystem $A$ or $B$, we can calculate the marginal probabilities:
\begin{align}
    P(0|\delta) &
    = \mathrm{Tr} \big(\rho_\text{out}^{(x)}\, \mathcal{M}_0\big) \nonumber
    = (1-C) 
    \sum_{p=0}^{\sigma} \left(\frac{1-C}{C}\right)^p |\mathcal{A}_\sigma(k,p)|^2 \hspace{1mm}
    \sum_{n=p}^{\infty} {n \choose p} C^n \bra{n}\mathcal{M}_0\ket{n}
    \\ & \overset{\ref{Local Measurement Matrix}}
    = (1-C) e^{-\eta_d |\delta|^2}
    \sum_{p=0}^{\sigma} \left(\frac{1-C}{C}\right)^p |\mathcal{A}_\sigma(k,p)|^2 \hspace{1mm} 
    \sum_{n=p}^{\infty} {n \choose p} \big( C \eta_d^2 |\delta|^{2} \big)^{n} 
    \sum_{i=0}^{n} \frac{\big(\frac{1-\eta_d}{\eta_d^2 |\delta|^{2}}\big)^i}{i!}\frac{n!}{(n-i)!^2}.\nonumber
\end{align}
For the special case of $\sigma=k=1$, corresponding to the 1-photon protocol, since $|\mathcal{A}_\sigma(k,p)|^2=\frac{1}{2}$, we have
\begin{align}
    P_\text{1ph}(0|\delta) &
    = \frac{1-C}{2} \hspace{1mm} e^{-\eta_d |\delta|^2}
    \sum_{i=0}^{\infty} \frac{\big(\frac{1-\eta_d}{\eta_d^2 |\delta|^{2}}\big)^i}{i!} \sum_{n=i}^{\infty} \frac{(C \eta_d^2 |\delta|^{2})^{n} n!}{(n-i)!^2}
    \sum_{p=0}^{1} {n \choose p} \left(\frac{1-C}{C}\right)^p \nonumber \\ &
    = \frac{1-C}{2} \hspace{1mm} e^{-\eta_d|\delta|^2}
    \sum_{i=0}^{\infty} \frac{\big(\frac{1-\eta_d}{\eta_d^2 |\delta|^{2}}\big)^i}{i!} \sum_{n=i}^{\infty} \frac{(C \eta_d^2 |\delta|^{2})^{n} n!}{(n-i)!(n-i)!} \left(1+n\left(\frac{1-C}{C}\right)\right).
\end{align}
Using Lemma~\ref{lemma:equations}.(b) we get
\begin{align}
    P_\text{1ph}(0|\delta) &
    =  \scalebox{1.2}{$\frac{1-C}{2}$} 
    \hspace{1mm} e^{(C\eta_d^2-\eta_d)|\delta|^2}
    \sum_{i=0}^{\infty} \frac{\big(C(1-\eta_d)\big)^i}{i!}
    \Bigg[ \scalebox{1.2}{$\left(1+\frac{1-C}{C}i\right)$}
    \sum_{j=0}^i \frac{i!}{j!} {i \choose j} (C\eta_d^2|\delta|^2)^j 
    + \scalebox{1.2}{$\frac{1-C}{C}$}
    \sum_{j=0}^i \frac{i!}{j!} {i+1 \choose j+1} (C\eta_d^2|\delta|^2)^{j+1} \Bigg] \nonumber \\ &
    =  \frac{1-C}{2} \ e^{(C\eta_d^2-\eta_d)|\delta|^2}
    \sum _{j=0}^{\infty} \frac{(C\eta_d^2|\delta|^2)^j}{j!}
    \sum_{i=j}^{\infty} \Bigg[ \scalebox{1.2}{$\frac{2C-1}{C}$} {i \choose j}
    + \Big( \scalebox{1.2}{$\frac{1-C}{C}$}(j+1)+ (1-C)\eta_d^2|\delta|^2 \Big) {i+1 \choose j+1}\Bigg]
    \big(C(1-\eta_d)\big)^i. \nonumber \\
\end{align}
Then, by using $\sum_{S=k}^{\infty} x^S {S \choose k} = x^k/(1-x)^{k+1}$ we obtain
\begin{align}
    P_\text{1ph}(0|\delta) &
    =  \frac{1-C}{2} \hspace{1mm} e^{(C\eta_d^2-\eta_d)|\delta|^2}
    \sum _{j=0}^{\infty} \frac{(C\eta_d^2|\delta|^2)^j}{j!}
    \Bigg[\scalebox{1.2}{$\frac{2C-1}{C}$} \frac{(C(1-\eta_d))^j}{(1-C(1-\eta_d))^{j+1}}
    + \scalebox{1.2}{$\frac{1-C}{C}$} \big(j+1 + C\eta_d^2 |\delta|^2\big) \frac{(C(1-\eta_d))^j}{(1-C(1-\eta_d))^{j+2}} \Bigg] \nonumber \\ &
    =  \scalebox{1.2}{$\frac{1-C}{2}$}
    \hspace{1mm} e^{(C\eta_d^2-\eta_d)|\delta|^2}
    \sum _{j=0}^{\infty} \frac{1}{j!}
    \Big(\frac{C^2(1-\eta_d)\eta_d^2|\delta|^2}{1-C(1-\eta_d)}\Big)^j
    \Bigg[ \scalebox{1.17}{$\frac{2C-1}{C(1-C(1-\eta_d))}$}
    +  \scalebox{1.17}{$\frac{(1-C)(1+C\eta_d^2|\delta|^2)}{C(1-C(1-\eta_d))^2}$}
    +  \scalebox{1.17}{$\frac{C(1-C)(1-\eta_d)\eta_d^2|\delta|^2}{(1-C(1-\eta_d))^3}$}
 \Bigg].
\end{align}
Finally, by applying the series expansion of $e^x$ and simplifying the expression, we get
\begin{equation}
    \boxed{P_\text{1ph}(0|\delta) = 
    \frac{1-C}{2}
    \Bigg[ \frac{2(1-C(1-\eta_d))^2 - \eta_d(1-C(1-\eta_d)) + (1-C)\eta_d^2|\delta|^2}{(1-C(1-\eta_d))^3}\Bigg]
    \exp{\Big(-\frac{\eta_d (1-C)|\delta|^2}{1-C(1-\eta_d)}\Big).}}
\end{equation}
In the lossless transmission scenario, $r=0$, by substituting $C=r_t \tanh^2 g =0 $, the expression reduces to:
\begin{equation}
    P_\text{1ph}(0|\delta) = \frac{2 -\eta_d + \eta_d^2 \lvert \delta \rvert^2}{2} e^{- \eta_d \lvert \delta \rvert^2}.
\end{equation}
For the perfect detection case ($\eta_d=1$), by replacing $C = r_t \tanh^2 g$ we obtain
\begin{align}
    P_\text{1ph}(0|\delta) = \frac{(1-r_t\tanh^2{g})}{2} \Big(1+ (1-r_t\tanh^2{g}) \lvert \delta \rvert^{2} \Big) e^{- (1-r_t\tanh^2{g}) \lvert \delta \rvert^{2}}.
\end{align}

\bigskip
\noindent
\textbf{II) Joint probabilities:}
To find the joint probability distribution for the state in \ref{rho_out_separate}, we need to insert it in
\begin{align}
    & P(00|\alpha \beta) 
    = \; \mathrm{Tr}\Big[\rho_{out}^{(k,\sigma)}\Big(D(\alpha) E_{0}^{\eta_{d}} D^{\dagger}(\alpha) \otimes D(\beta) E_{0}^{\eta_{d}} D^{\dagger}(\beta)\Big) \Big]
    = (1-C)^{\sigma+2} \sum_{p,q=0}^{\sigma} \mathcal{A}_\sigma(k,p) \mathcal{A}^*_\sigma(k,q) \\
    & \times \Big[ \sum_{m=0}^{\infty} C^m 
    \scalebox{1.25}{$\sqrt{\binom{m+p}{m}\binom{m+q}{m}}$} \bra{m+q}\mathcal{M}^0_A\ket{m+p} \Big] 
    \Big[ \sum_{n=0}^{\infty} C^n 
    \scalebox{1.25}{$\sqrt{\binom{n+\sigma-p}{n}\binom{n+\sigma-q}{n}}$} \bra{n+\sigma-q}\mathcal{M}^0_B\ket{n+\sigma-p} \Big] \nonumber.
\end{align}
Using \eqref{Local Measurement Matrix}, we obtain
\begin{align}
    &P(00|\alpha \beta) = 
    e^{-\eta_d (|\alpha|^2+|\beta|^2)} \eta_d^{2\sigma}(1-C)^{\sigma+2}
    \sum_{p,q=0}^{\sigma} 
    \frac{(\alpha^*)^{p} (\beta^*)^{\sigma-p} \mathcal{A}_\sigma(k,p)}{\sqrt{p! (\sigma-p)!}} .
    \frac{\alpha^q \beta^{\sigma-q} \mathcal{A}^*_\sigma(k,q)}{\sqrt{q!(\sigma-q)!}}\\
    &\times \Big[
     \sum_{i=0}^\infty \frac{(\frac{1-\eta_d}{\eta_d^2 |\alpha|^2})^{i}}{i!}
    \sum_{\substack{m \geq 0 \\ m\geq i-p \\ m \geq i-q}}^{\infty} \frac{(C\eta_d^2 |\alpha|^2)^{m} (m+p)!(m+q)!}{m! (m+p-i)!(m+q-i)!}\Big]
    \Big[ \sum_{j=0}^\infty \frac{(\frac{1-\eta_d}{\eta_d^2 |\beta|^2})^{j}}{j!}
    \sum_{\substack{n \geq 0 \\ n \geq j+p-\sigma \\ n\geq j+q-\sigma}}^{\infty} \frac{(C\eta_d^2 |\beta|^2)^{n} (n+\sigma-p)!(n+\sigma-q)!}{n! (n+\sigma-p-j)!(n+\sigma-q-j)!}\Big].\nonumber
\end{align}
For the 1-photon protocol, by setting $\sigma=k=1$, the expression for $p(00|\alpha\beta)$  includes four terms corresponding to the possible pairs $(p,q) \in \{(0,0), (0,1), (1,0), (1,1)\}$:
\begin{align}
    &P_\text{1ph}(00|\alpha \beta) = 
    e^{-\eta_d (|\alpha|^2+|\beta|^2)} \frac{\eta_d^{2}(1-C)^{3}}{2} \times \nonumber\\
    & \Bigg[|\beta|^2
    \sum_{i=0}^\infty \frac{(\frac{1-\eta_d}{\eta_d^2 |\alpha|^2})^{i}}{i!}
    \sum_{m=i}^{\infty} \frac{(C\eta_d^2 |\alpha|^2)^{m}\, m!}{ (m-i)!(m-i)!} .
    \sum_{j=0}^\infty \frac{(\frac{1-\eta_d}{\eta_d^2 |\beta|^2})^{j}}{j!}
    \sum_{n=j}^{\infty} \frac{(C\eta_d^2 |\beta|^2)^{n-1} \ n. n!}{ (n-j)!(n-j)!} \nonumber\\
    & - i \alpha \beta^*
    \sum_{i=0}^\infty \frac{(\frac{1-\eta_d}{\eta_d^2 |\alpha|^2})^{i}}{i!}
    \sum_{m=i}^{\infty} \frac{(C\eta_d^2 |\alpha|^2)^{m} (m+1)!}{(m-i)!(m+1-i)!}.
    \sum_{j=0}^\infty \frac{(\frac{1-\eta_d}{\eta_d^2 |\beta|^2})^{j}}{j!}
    \sum_{n=j}^{\infty} \frac{(C\eta_d^2 |\beta|^2)^{n} (n+1)!}{(n+1-j)!(n-j)!} \nonumber\\
    & + i\alpha^* \beta
     \sum_{i=0}^\infty \frac{(\frac{1-\eta_d}{\eta_d^2 |\alpha|^2})^{i}}{i!}
    \sum_{m=i}^{\infty} \frac{(C\eta_d^2 |\alpha|^2)^{m} (m+1)!}{(m+1-i)!(m-i)!} .
    \sum_{j=0}^\infty \frac{(\frac{1-\eta_d}{\eta_d^2 |\beta|^2})^{j}}{j!}
    \sum_{n=j}^{\infty} \frac{(C\eta_d^2 |\beta|^2)^{n} (n+1)!}{(n-j)!(n+1-j)!} \nonumber \\
    & + |\alpha|^2  
    \sum_{i=0}^\infty \frac{(\frac{1-\eta_d}{\eta_d^2 |\alpha|^2})^{i}}{i!}
    \sum_{m=i}^{\infty} \frac{(C\eta_d^2 |\alpha|^2)^{m-1} \, m.m!}{ (m-i)!(m-i)!} .
    \sum_{j=0}^\infty \frac{(\frac{1-\eta_d}{\eta_d^2 |\beta|^2})^{j}}{j!}
    \sum_{n=j}^{\infty} \frac{(C\eta_d^2 |\beta|^2)^{n} \, n!}{(n-j)!(n-j)!}\Bigg].
\end{align}
By applying Lemma~\ref{lemma:equations}.(a,b) to the summations over $m$ and $n$, we obtain
\begin{align}
    & P_\text{1ph}(00|\alpha \beta) = 
    e^{(C\eta_d^2-\eta_d) (|\alpha|^2+|\beta|^2)} \times\frac{\eta_d^{2}(1-C)^{3}}{2} \times \nonumber\\
    & \Bigg[ |\alpha|^2 
    \sum_{i=0}^\infty \frac{(\frac{1-\eta_d}{\eta_d^2 |\alpha|^2})^{i}}{i!}
    \sum_{u=0}^{i} \frac{i!}{u!} \Big[ \binom{i+1}{u+1}(C\eta_d^2 |\alpha|^2)^{i+u} + i\binom{i}{u} (C\eta_d^2 |\alpha|^2)^{i+u-1}\Big].
    \sum_{j=0}^\infty \frac{(\frac{1-\eta_d}{\eta_d^2 |\beta|^2})^{j}}{j!}
    \sum_{v=0}^{j} \binom{j}{v}\frac{j!(C\eta_d^2 |\beta|^2)^{j+v}}{v!}\nonumber\\
    & + |\beta|^2 
    \sum_{i=0}^\infty \frac{(\frac{1-\eta_d}{\eta_d^2 |\alpha|^2})^{i}}{i!}
    \sum_{u=0}^{i} \binom{i}{u}\frac{i!(C\eta_d^2 |\alpha|^2)^{i+u}}{u!}.
    \sum_{j=0}^\infty \frac{(\frac{1-\eta_d}{\eta_d^2 |\beta|^2})^{j}}{j!}
    \sum_{v=0}^{j} \frac{j!}{v!} \Big[ \binom{j+1}{v+1}(C\eta_d^2 |\beta|^2)^{j+v} + j\binom{j}{v} (C\eta_d^2 |\beta|^2)^{j+v-1}\Big]\nonumber\\
    & + 2 \text{Im}(\alpha \beta^*) 
    \sum_{i=0}^\infty \frac{(\frac{1-\eta_d}{\eta_d^2 |\alpha|^2})^{i}}{i!}
    \sum_{u=0}^{i} \binom{i+1}{u+1}\frac{i!(C\eta_d^2 |\alpha|^2)^{i+u}}{u!}.
    \sum_{j=0}^\infty \frac{(\frac{1-\eta_d}{\eta_d^2 |\beta|^2})^{j}}{j!}
    \sum_{v=1}^{j} \binom{j+1}{v+1}\frac{j!(C\eta_d^2 |\beta|^2)^{j+v}}{v!}
    \Bigg].
\end{align}
where $\text{Im}(x)$ returns the imaginary part of a complex number $x$.
Changing the order of summation yields us
\begin{align}
    & P_\text{1ph}(00|\alpha \beta) = 
    e^{(C\eta_d^2-\eta_d) (|\alpha|^2+|\beta|^2)} \times\frac{\eta_d^{2}(1-C)^{3}}{2} \times \nonumber\\
    & \Bigg[ |\alpha|^2 
    \sum_{u=0}^\infty \frac{(C\eta_d^2 |\alpha|^2)^{u}}{u!}
    \sum_{i=u}^{\infty} (C(1-\eta_d))^{i} \Big[ \binom{i+1}{u+1} + \frac{i\binom{i}{u}} {C\eta_d^2 |\alpha|^2}\Big].
    \sum_{v=0}^\infty \frac{(C\eta_d^2 |\beta|^2)^{v}}{v!}
    \sum_{j=v}^{\infty} \binom{j}{v}(C(1-\eta_d))^{j}\nonumber\\
    & + |\beta|^2 
    \sum_{u=0}^\infty \frac{(C\eta_d^2 |\alpha|^2)^{u}}{u!}
    \sum_{i=u}^{\infty} \binom{i}{u}(C(1-\eta_d))^{i}.
    \sum_{v=0}^\infty \frac{(C\eta_d^2 |\beta|^2)^{v}}{v}
    \sum_{j=v}^{\infty} (C(1-\eta_d))^{j} \Big[ \binom{j+1}{v+1} + \frac{j\binom{j}{v}} {C\eta_d^2 |\beta|^2}\Big] \nonumber\\
    & + 2 \ \text{Im}(\alpha \beta^*) 
    \sum_{u=0}^\infty \frac{(C\eta_d^2 |\alpha|^2)^{u}}{u!}
    \sum_{i=u}^{\infty} \binom{i+1}{u+1}(C(1-\eta_d))^{i}.
    \sum_{v=0}^\infty \frac{(C\eta_d^2 |\beta|^2)^{v}}{v!}
    \sum_{j=v}^{\infty} \binom{j+1}{v+1}(C(1-\eta_d))^{j}
    \Bigg]
\end{align}
Next, we use Lemma~\ref{lemma:equations}.(c,d) for summations over $i$ and $j$:
\begin{align}
    & P_\text{1ph}(00|\alpha \beta) = 
    e^{(C\eta_d^2-\eta_d) (|\alpha|^2+|\beta|^2)} \times\frac{\eta_d^{2}(1-C)^{3}}{2} \times \nonumber\\
    & \Bigg[ |\alpha|^2 
    \sum_{u=0}^\infty \frac{(C\eta_d^2 |\alpha|^2)^{u}}{u!}
    \Big[ \frac{\big(C(1-\eta_d)\big)^u}{\big(1-C(1-\eta_d)\big)^{u+2}} 
    + \frac{\big(C(1-\eta_d)\big)^{u+1} + u \big(C(1-\eta_d)\big)^{u}}{C\eta_d^2|\alpha|^2 \big(1-C(1-\eta_d)\big)^{u+2}} \Big].
    \sum_{v=0}^\infty \frac{(C\eta_d^2 |\beta|^2)^{v}}{v!}
    \frac{\big(C(1-\eta_d)\big)^v}{\big(1-C(1-\eta_d)\big)^{v+1}} \nonumber\\
    & + |\beta|^2 
    \sum_{u=0}^\infty \frac{(C\eta_d^2 |\alpha|^2)^{u}}{u!}
    \frac{\big(C(1-\eta_d)\big)^u}{\big(1-C(1-\eta_d)\big)^{u+1}}.
    \sum_{v=0}^\infty \frac{(C\eta_d^2 |\beta|^2)^{v}}{v!}
    \Big[ \frac{\big(C(1-\eta_d)\big)^v}{\big(1-C(1-\eta_d)\big)^{v+2}} 
    + \frac{\big(C(1-\eta_d)\big)^{v+1} + v \big(C(1-\eta_d)\big)^{v}}{C\eta_d^2|\beta|^2 \big(1-C(1-\eta_d)\big)^{v+2}} \Big] \nonumber\\
    & + 2 \ \text{Im}(\alpha \beta^*) 
    \sum_{u=0}^\infty \frac{(C\eta_d^2 |\alpha|^2)^{u}}{u!}
    \frac{\big(C(1-\eta_d)\big)^u}{\big(1-C(1-\eta_d)\big)^{u+2}}.
    \sum_{v=0}^\infty \frac{(C\eta_d^2 |\beta|^2)^{v}}{v!}
    \frac{\big(C(1-\eta_d)\big)^v}{\big(1-C(1-\eta_d)\big)^{v+2}}
    \Bigg].
\end{align}
Using the Taylor series expansion for $e^x=\sum\limits_{n=0}^\infty \frac{x^n}{n!}$, the joint probability simplifies to
\begin{align}
    \boxed{P_\text{1ph}(00|\alpha,\beta) 
    = \frac{(1 - C)^3}{2} \left( 
    \frac{2(1 - \eta_d)(1 - C(1 - \eta_d)) + \eta_d^2 \big(|\alpha|^2 + |\beta|^2 + 2\text{Im}(\alpha\beta^*)\big)}{ (1-C(1-\eta_d))^4} \right)
    \exp\left( - \frac{ \eta_d (1 - C) (|\alpha|^2 + |\beta|^2)}{1-C(1-\eta_d)} \right).} 
\end{align}
\par\vspace*{-4mm}
\noindent
In the lossless transmission scenario ($r=0$), by substituting $C = r_t \tanh^2 g =0 $, the expression reduces to

\vspace{-2mm}

\begin{align}
    P_\text{1ph}(00|\alpha \beta) = \frac{1}{2} e^{-\eta_d (\lvert \alpha \rvert^2 + \lvert \beta \rvert^2)}\Big[2 - 2\eta_d + \eta_d^{2} \big( \lvert \alpha \rvert^2 + \lvert \beta \rvert^2 + 2 \mathrm{Im}(\alpha \beta^{*} ) \big) \Big].
\end{align}
For the perfect detection case ($\eta_d=1$), replacing $C = r_t \tanh^2 g$, we obtain
\begin{align}
    P_\text{1ph}(00|{\alpha} {\beta}) = \frac{(1-r_t\tanh^2{g})^{3} }{2} 
    \Big[ \lvert {\alpha} \rvert^{2} + \lvert {\beta} \rvert^{2} + 2 \mathrm{Im}({\alpha} {\beta}^{*} ) \Big]
    e^{-(\lvert {\alpha} \rvert^{2}+\lvert {\beta} \rvert^{2})(1-r_t\tanh^2{g})}. 
\end{align}

\subsection{The 2-photon protocol}
\label{App:Imperfect-lossy-2ph}
For the 2-photon protocol, the calculations closely follow those of the 1-photon case.

\medskip
\noindent
\textbf{I) Marginals:} Since the state is not symmetric in the 2-photon protocol, Alice's marginal distribution differs from Bob's. 
To calculate these distributions, we first determine the reduced density matrices for Alice and Bob by taking the partial trace of the state in Eq.~\ref{rho_out_2ph}
\begin{align}\label{reduced_rho_out_2ph}
    \rho_\text{out}^{(a)} = & \, \frac{{(\mathcal{N N}_b)}^{-2} \sqrt{\eta_t}}{\cosh^2{g_b}}
    \Bigg[ t_c\sum\limits_{m=0}^\infty \lambda^{(a)}_m r_t^m \ketbra{m}{m}_{a_1} 
    .\frac{\partial}{\partial r_t} 
    + r_c \frac{\partial}{\partial r_t} \big(\sum\limits_{m=0}^\infty \lambda^{(a)}_m \, r_t^{m} \ketbra{m}{m}_{a_1}\big) \Bigg] \nonumber\\
    &\hspace{2cm} \times  \Big(\frac{\cosh^{-2}{g_b}}{1-(r_s + t_s r_t) \tanh^2{g_b}} - \frac{\cosh^{-2}{g_b}}{1-(1-\eta_s)(r_s + t_s r_t) \tanh^2{g_b}}\Big), \\
    \rho_\text{out}^{(b)} = & \ {(\mathcal{N N}_b)}^{-2} \sqrt{\eta_t} \,
    \sum\limits_{n=0}^\infty \lambda^{(b)}_n \big (1-(1-\eta_s)^n \big) \nonumber\\
    \ \times & \Bigg[\frac{t_c}{1+\sqrt{\eta_t}\,\sinh^2{g_a}}. 
    \frac{\partial}{\partial r_t} \Big(\sum\limits_{l=0}^{n} \binom{n}{l} r_s^{l} (r_t t_s)^{n-l} \ketbra{l}{l}_{b_1}\Big) 
    + \frac{r_c \sinh^2{g_a}}{(1+\sqrt{\eta_t}\,\sinh^2{g_a})^2} .
    \sum\limits_{l=0}^{n} \binom{n}{l} r_s^{l} (r_t t_s)^{n-l} \ketbra{l}{l}_{b_1}
    \Bigg]
\end{align}
Using a similar calculation as in the 1-photon protocol, with the symmetric BS in Charlie's ($t_c=r_c=0.50$), the marginals can be obtained as follows
\begin{align}\label{Marge_a_before}
    & P_\text{2ph}(0|\alpha)
    = \ \mathrm{Tr}(\rho_\text{out}^{(a)} \mathcal{M}_0^{(A)})
    =  \frac{1}{2} {(\mathcal{N N}_b)}^{-2} \sqrt{\eta_t}
    \Bigg[\sum\limits_{m=0}^\infty \lambda^{(a)}_m r_t^m \bra{m} \mathcal{M}_0^{(A)} \ket{m} 
    .\frac{\partial}{\partial r_t} 
    + \frac{\partial}{\partial r_t} \big(\sum\limits_{m=0}^\infty \lambda^{(a)}_m \, r_t^{m} \bra{m} \mathcal{M}_0^{(A)} \ket{m}\big) \Bigg] \nonumber\\
     & \hspace{6.7cm} \times \Big(\frac{\cosh^{-2}{g_b}}{1-(r_s + t_s r_t) \tanh^2{g_b}} - \frac{\cosh^{-2}{g_b}}{1-(1-\eta_s)(r_s + t_s r_t) \tanh^2{g_b}}\Big),
\end{align}
\begin{align}\label{Marge_b_before}
    & P_\text{2ph}(0|\beta)
    = \ \mathrm{Tr}(\rho_\text{out}^{(b)} \mathcal{M}_0^{(B)}) 
    = \frac{1}{2} {(\mathcal{N N}_b)}^{-2} \sqrt{\eta_t} \,
    \sum\limits_{n=0}^\infty \lambda^{(b)}_n \big (1-(1-\eta_s)^n \big) \\
    & \times \Bigg[\frac{1}{1+\sqrt{\eta_t}\,\sinh^2{g_a}}. 
    \frac{\partial}{\partial r_t} \Big(\sum\limits_{l=0}^{n} \binom{n}{l} r_s^{l} (t_s r_t)^{n-l} \bra{l} \mathcal{M}_0^{(B)} \ket{l}\Big) 
    + \frac{\sinh^2{g_a}}{(1+\sqrt{\eta_t}\,\sinh^2{g_a})^2} .
    \sum\limits_{l=0}^{n} \binom{n}{l} r_s^{l} (t_s r_t)^{n-l} \bra{l}\mathcal{M}_0^{(B)}\ket{l}
    \Bigg].\nonumber
\end{align}
Using Eq.~\eqref{Local Measurement Matrix}, 
$ \bra{m} \mathcal{M}_0 \ket{m} = e^{-\eta_d |\delta|^2} m! \ (\eta_d|\delta|)^{2m}
\sum\limits_{k=0}^{m} \frac{\big(\frac{1-\eta_d}{\eta_d^2 |\delta|^2}\big)^{k}}{k!(m-k)!^2}$, and defining $C_{t,a(b)}:=r_t \tanh^2{g_{a(b)}} \ ,\  C_{s,a(b)}:=r_s \tanh^2{g_{a(b)}}$, we get

\begin{align} \label{Marg_a_simp}
    \hspace{-2.5cm}\sum\limits_{m=0}^\infty \lambda^{(a)}_m r_t^m \bra{m} \mathcal{M}_0^{(A)} \ket{m}
    & = \cosh^{\!-2}{g_a} \ e^{-\eta_d|\alpha|^2} \sum\limits_{k=0}^{\infty} \frac{(\frac{1-\eta_d}{\eta_d^2 |\alpha|^2})^{k}}{k!}
    \sum\limits_{m=k}^\infty \frac{m! \, (C_{t,a}\eta_d^2|\alpha|^2)^{m}}{(m-k)!^2}
    \nonumber\\
    & \overset{(a)}{=} \cosh^{\!-2}{g_a}\ e^{-\eta_d|\alpha|^2} e^{C_{t,a}\eta_d^2|\alpha|^2}
    \sum_{i=0}^\infty \frac{(C_{t,a}\eta_d^2|\alpha|^2)^{i}}{i!}
    \sum\limits_{k=i}^{\infty} \binom{k}{i} (C_{t,a}(1-\eta_d))^k
    \nonumber\\
    & \overset{(c)}{=} \cosh^{\!-2}{g_a}\ e^{-\eta_d|\alpha|^2} e^{C_{t,a}\eta_d^2|\alpha|^2} 
    \sum_{i=0}^\infty \frac{(C_{t,a}\eta_d^2|\alpha|^2)^{i}}{i!}
    \frac{(C_{t,a}(1-\eta_d))^i}{(1-C_{t,a}(1-\eta_d))^{i+1}}
    \nonumber\\
    & = \frac{\cosh^{\!-2}{g_a}}{1-C_{t,a}(1-\eta_d)}
    \exp{\Big( -\frac{\eta_d (1-C_{t,a}) |\alpha|^2}{1-C_{t,a}(1-\eta_d)}\Big)},
\end{align}
\par\vspace*{-3mm}
\begin{align}\label{Marg_b_simp}
     \sum\limits_{n=0}^\infty \ \lambda^{(b)}_n & \big (1-(1-\eta_s)^n \big) 
    \sum\limits_{l=0}^{n} \binom{n}{l} r_s^{l} (t_s r_t)^{n-l} \bra{l} \mathcal{M}_0^{(B)} \ket{l} 
    \nonumber\\
    & \overset{\ref{Local Measurement Matrix}}{=} \cosh^{\!-2}{g_b}\ e^{-\eta_d |\beta|^2} 
    \sum\limits_{k=0}^{\infty} \frac{(\frac{1-\eta_d}{\eta_d^2 |\beta|^2})^{k}}{k!} 
    \sum\limits_{l=k}^{\infty} \frac{l! \,(\frac{r_s \eta_d^2 |\beta|^2}{t_s r_t})^l}{(l-k)!^2}
    \sum\limits_{n=l}^\infty \binom{n}{l} (C_{t,b} t_s)^{n} \big (1-(1-\eta_s)^n \big) \nonumber\\
    & \overset{(c)}{=} \cosh^{\!-2}{g_b}\ e^{-\eta_d |\beta|^2} 
    \sum\limits_{k=0}^{\infty} \frac{(\frac{1-\eta_d}{\eta_d^2 |\beta|^2})^{k}}{k!} 
    \sum\limits_{l=k}^{\infty} \frac{l! \, (\frac{r_s \eta_d^2 |\beta|^2}{t_s r_t})^l}{(l-k)!^2}
    \left[ \frac{(C_{t,b} t_s)^{l}}{(1-C_{t,b} t_s)^{l+1}}
    - \frac{\big(C_{t,b} t_s(1-\eta_s)\big)^{l}}{\big(1-C_{t,b} t_s(1-\eta_s)\big)^{l+1}} \right]\nonumber\\
    & \overset{(a)}{=} \cosh^{\!-2}{g_b}\ e^{-\eta_d |\beta|^2} 
    \Bigg[ \frac{e^{\frac{C_{s,b} \ \eta_d^2 |\beta|^2}{1-C_{t,b} t_s}}}{1-C_{t,b} t_s} 
    \sum_{i=0}^\infty \frac{\Big(\frac{C_{s,b} \ \eta_d^2 |\beta|^2}{1-C_{t,b} t_s}\Big)^i}{i!}
    \sum\limits_{k=i}^{\infty} \binom{k}{i} \Big(\frac{C_{s,b} (1-\eta_d)}{1-C_{t,b} t_s}\Big)^{k} 
    \nonumber\\
    & \hspace{3cm} - \frac{e^{\frac{C_{s,b} \ \eta_d^2 |\beta|^2(1-\eta_s)}{1-C_{t,b} t_s(1-\eta_s)}}}{1-C_{t,b} t_s(1-\eta_s)}
    \sum_{i=0}^\infty \frac{\Big(\frac{C_{s,b} \ \eta_d^2 |\beta|^2 (1-\eta_s)}{1-C_{t,b} t_s (1-\eta_s)}\Big)^i}{i!}
    \sum\limits_{k=i}^{\infty} \binom{k}{i} \Big(\frac{C_{s,b} (1-\eta_d) (1-\eta_s)}{1-C_{t,b} t_s (1-\eta_s)}\Big)^{k} \Bigg] 
    \nonumber\\
    & \overset{(c)}{=} \cosh^{\!-2}{g_b}\ e^{-\eta_d |\beta|^2} 
    \Bigg[ \frac{e^{\frac{C_{s,b} \ \eta_d^2 |\beta|^2}{1-C_{t,b} t_s}}}{1-C_{t,b} t_s - C_{s,b} (1-\eta_d)} 
    \sum_{i=0}^\infty \frac{\Big(\frac{C_{s,b}^2 (1-\eta_d)\eta_d^2 |\beta|^2}{(1-C_{t,b} t_s) (1-C_{t,b} t_s - C_{s,b} (1-\eta_d))}\Big)^i}{i!}
    \nonumber\\
    & \hspace{3cm} - \frac{e^{\frac{C_{s,b} \ \eta_d^2 |\beta|^2(1-\eta_s)}{1-C_{t,b} t_s(1-\eta_s)}}}{ 1-(C_{t,b} t_s+C_{s,b} (1-\eta_d)) (1-\eta_s)}
    \sum_{i=0}^\infty \frac{\Big(\frac{C_{s,b}^2 (1-\eta_d) \eta_d^2 |\beta|^2 (1-\eta_s)^2}{(1-C_{t,b} t_s (1-\eta_s)) (1-(C_{t,b} t_s+C_{s,b} (1-\eta_d)) (1-\eta_s))}\Big)^i}{i!}\Bigg] 
    \nonumber\\
    & = \cosh^{\!-2}{g_b}
    \Bigg[ \frac{\exp{\Big(-\frac{\eta_d (1-C_{t,b} t_s - C_{s,b}) |\beta|^2}{1-C_{t,b} t_s - C_{s,b} (1-\eta_d)}\Big)}}
    {1-C_{t,b} t_s - C_{s,b} (1-\eta_d)}
    - \frac{\exp{\Big(-\frac{\eta_d \big(1- (C_{t,b} t_s+C_{s,b})(1-\eta_s)\big) |\beta|^2 }{1-(C_{t,b} t_s+C_{s,b} (1-\eta_d)) (1-\eta_s)}\Big)}}
    {1-(C_{t,b} t_s+C_{s,b} (1-\eta_d)) (1-\eta_s)}\Bigg]
\end{align}
where we applied parts (a) and (c) of Lemma~\ref{lemma:equations}, and used the Taylor expansion of $e^x$ to simplify the expression. Substituting Eqs.~\eqref{Marg_a_simp} and \eqref{Marg_b_simp} into Eqs.~\eqref{Marge_a_before} and \eqref{Marge_b_before}, then yields the marginals $P_\text{2ph}(0|\alpha)$ and $P_\text{2ph}(0|\beta)$, respectively.
In the limit $g_b\ll 1$, where Bob's initial state approaches the SPPE state, the marginals is given by
\begin{align}
    \boxed{\tilde{P}_\text{2ph}(0|\alpha)
    = \exp{\Big(-\frac{(1-C_{t,a})\eta_d|\alpha|^2}{1-C_{t,a}(1-\eta_d)}\ \Big)}
    \!\times\!\left(\frac{1 - C_{t,a}}{1 - C_{t,a}(1 - \eta_d)}\right)^{\!2}
    \!\times\frac{ t_s \Bigl(1+\frac{C_{t,a}\eta_d^2 |\alpha|^2}{1 - C_{t,a} (1 -\eta_d)} \Bigr)
    +C_{s,a}\Bigl(1-\eta_d+\frac{\eta_d^2 |\alpha|^2}{1 - C_{t,a} (1 -\eta_d)}\Bigr)}
    {t_s + r_s \tanh^2{g_a}}}\nonumber
\end{align}
\vspace{-3mm}
\begin{align}
    \boxed{\tilde{P}_\text{2ph}(0|\beta) = 
    e^{-\eta_d |\beta|^2} \times
    \frac{ t_s + C_{s,a} \left(1-\eta_d +\eta_d^2 |\beta|^2\right) }
    {t_s + r_s \tanh^2{g_a}}} \nonumber 
\end{align}

\noindent
\textbf{II) Joint probabilities:}
The joint probability distribution corresponding to the state in Eq.~\eqref{rho_out_2ph} is obtained from
$\mathrm{Tr}\big[\rho_{out}^{(2ph)} \big(\mathcal{M}_0^{(A)} \otimes\mathcal{M}_0^{(B)}\big) \big]$. 
For a symmetric BS at Charlie’s ($t_c=r_c=0.50$), this probability is calculated as:

    \begin{align}
    P_\text{2ph}(00|\alpha,\beta) & = \frac{1}{2}{(\mathcal{N N}_b)}^{-2} \sqrt{\eta_t}
    \sum\nolimits_{n=0}^\infty \lambda^{(b)}_n \big (1-(1-\eta_s)^n \big)
    \nonumber\\
    & \times \Bigg[ \sum\limits_{m=0}^\infty \lambda^{(a)}_m r_t^m \bra{m}\mathcal{M}_0^{(A)}\ket{m}
    .\frac{\partial}{\partial r_t} \Big(\sum\limits_{l=0}^{n} \binom{n}{l} r_s^{l} (r_t t_s)^{n-l} \bra{l}\mathcal{M}_0^{B}\ket{l}\Big)
    \nonumber\\
    & + \frac{\partial}{\partial r_t} \Big(\sum\limits_{m=0}^\infty \lambda^{(a)}_m r_t^m \bra{m}\mathcal{M}_0^{(A)}\ket{m}\Big)
    . \sum\limits_{l=0}^{n} \binom{n}{l} r_s^{l} (r_t t_s)^{n-l} \bra{l}\mathcal{M}_0^{B}\ket{l}
    \nonumber\\
    & + 2 \sqrt{\frac{t_s}{r_s}}\text{Re}\Big(e^{i\phi } 
    \sum\limits_{m=0}^\infty \sqrt{\lambda^{(a)}_m \lambda^{(a)}_{m+1} (m+1)} \ r_t^m \bra{m+1}\mathcal{M}_0^{(A)}\ket{m}
    . \sum_{l=1}^n \sqrt{l}\binom{n}{l} r_s^l (t_s r_t)^{n-l} \bra{l}\mathcal{M}_0^{(B)}\ket{l-1}\Big)\Bigg]
\end{align}
which consists of three terms: the first two can be obtained from Eqs.~\eqref{Marg_a_simp} and \eqref{Marg_b_simp}, while for the last term we have
\begin{align}
    \sum\limits_{m=0}^\infty \sqrt{\lambda^{(a)}_m \lambda^{(a)}_{m+1} (m+1)} \ r_t^m \bra{m+1}\mathcal{M}_0^{(A)}\ket{m} \!
    & \overset{\ref{Local Measurement Matrix}}{=} 
    \frac{\eta_d \alpha \tanh{g_a}}{\cosh^2{g_a}} e^{-\eta_d|\alpha|^2}
    \sum_{i=0}^\infty \frac{(\frac{1-\eta_d}{\eta_d^2|\alpha|^2})^i}{i!} 
    \sum_{m=i}^\infty \frac{(m+1)! (C_{t,a}\eta_d^2|\alpha|^2)^m}{(m+1-i)!(m-i)!}
    \nonumber\\
    & \! \overset{(a)}{=} \!
    \frac{\eta_d \alpha \tanh{g_a}}{\cosh^2{g_a}} e^{-\eta_d|\alpha|^2} e^{C_{t,a}\eta_d^2|\alpha|^2}
    \sum_{i=0}^\infty \frac{(\frac{1-\eta_d}{\eta_d^2|\alpha|^2})^i}{i!} 
    \sum_{j=0}^{i} \frac{i!(i+1)!(C_{t,a}\eta_d^2|\alpha|^2)^{i+j}}{j!(j+1)!(i-j)!} 
    \nonumber\\
    & \! = \frac{\eta_d \alpha \tanh{g_a}}{\cosh^2{g_a}} e^{(C_{t,a}\eta_d^2-\eta_d)|\alpha|^2}
    \sum_{j=0}^\infty \frac{(C_{t,a}\eta_d^2|\alpha|^2)^j}{j!}
    \sum_{i=j}^\infty \binom{i+1}{j+1} (C_{t,a}(1-\eta_d))^ i
    \nonumber\\
    & \! \overset{(c)}{=} \frac{\eta_d \alpha \tanh{g_a}}{\cosh^2{g_a}} e^{(C_{t,a}\eta_d^2-\eta_d)|\alpha|^2}
    \sum_{j=0}^\infty \frac{(C_{t,a}\eta_d^2|\alpha|^2)^j}{j!}
    \frac{\big(C_{t,a}(1-\eta_d)\big)^{j}}{\big(1-C_{t,a}(1-\eta_d)\big)^{j+2}}
    \nonumber\\
    & \! = \frac{\eta_d \alpha \tanh{g_a} \cosh^{-2}{g_a}}{\big(1-C_{t,a}(1-\eta_d)\big)^{2}}
    \exp{\Big(-\frac{\eta_d (1-C_{t,a})|\alpha|^2}{1-C_{t,a}(1-\eta_d)}\Big)},
\end{align}
\begin{align}
    \sum\limits_{n=0}^\infty \lambda^{(b)}_n & \big (1-(1-\eta_s)^n \big)
    \sum_{l=1}^n \sqrt{l}\binom{n}{l} r_s^l (t_s r_t)^{n-l} \bra{l}\mathcal{M}_0^{(B)}\ket{l-1}\\
    & \overset{\ref{Local Measurement Matrix}}{=}
    \frac{e^{-\eta_d |\beta|^2}}{\eta_d\beta^*\cosh^{2}{g_b}} 
    \sum\limits_{k=0}^{\infty} \frac{(\frac{1-\eta_d}{\eta_d^2 |\beta|^2})^{k}}{k!} 
    \sum\limits_{l=k+1}^{\infty} \frac{l! \,(\frac{r_s \eta_d^2 |\beta|^2}{t_s r_t})^l}{(l-k)!(l-k-1)!}
    \sum\limits_{n=l}^\infty \binom{n}{l} (C_{t,b} t_s)^{n} \big (1-(1-\eta_s)^n \big) \nonumber\\
    & \overset{(c)}{=} \frac{e^{-\eta_d |\beta|^2}}{\eta_d\beta^*\cosh^{2}{g_b}} 
    \sum\limits_{k=0}^{\infty} \frac{(\frac{1-\eta_d}{\eta_d^2 |\beta|^2})^{k}}{k!} 
    \sum\limits_{l=k}^{\infty} \frac{l! \, (\frac{r_s \eta_d^2 |\beta|^2}{t_s r_t})^l}{(l-k)!(l-k-1)!}
    \left[ \frac{(C_{t,b} t_s)^{l}}{(1-C_{t,b} t_s)^{l+1}}
    - \frac{\big(C_{t,b} t_s(1-\eta_s)\big)^{l}}{\big(1-C_{t,b} t_s(1-\eta_s)\big)^{l+1}} \right]\nonumber\\
    & \overset{(a)}{=} \frac{\eta_d\beta \, e^{-\eta_d |\beta|^2}}{(1-\eta_d)\cosh^{2}{g_b}} 
    \Bigg[ \frac{e^{\frac{C_{s,b} \ \eta_d^2 |\beta|^2}{1-C_{t,b} t_s}}}{1-C_{t,b} t_s} 
    \sum_{i=0}^\infty \frac{\Big(\frac{C_{s,b} \ \eta_d^2 |\beta|^2}{1-C_{t,b} t_s}\Big)^i}{i!}
    \sum\limits_{k=i}^{\infty} \binom{k+1}{i+1} \Big(\frac{C_{s,b} (1-\eta_d)}{1-C_{t,b} t_s}\Big)^{k+1} 
    \nonumber\\
    & \hspace{3cm} - \frac{e^{\frac{C_{s,b} \ \eta_d^2 |\beta|^2(1-\eta_s)}{1-C_{t,b} t_s(1-\eta_s)}}}{1-C_{t,b} t_s(1-\eta_s)}
    \sum_{i=0}^\infty \frac{\Big(\frac{C_{s,b} \ \eta_d^2 |\beta|^2 (1-\eta_s)}{1-C_{t,b} t_s (1-\eta_s)}\Big)^i}{i!}
    \sum\limits_{k=i}^{\infty} \binom{k+1}{i+1} \Big(\frac{C_{s,b} (1-\eta_d) (1-\eta_s)}{1-C_{t,b} t_s (1-\eta_s)}\Big)^{k+1} \Bigg] 
    \nonumber\\
    & \overset{(c)}{=} \frac{C_{s,b}\eta_d\beta \, e^{-\eta_d |\beta|^2}}{\cosh^{2}{g_b}} 
    \Bigg[ \frac{e^{\frac{C_{s,b} \ \eta_d^2 |\beta|^2}{1-C_{t,b} t_s}}}{(1-C_{t,b} t_s - C_{s,b} (1-\eta_d))^2} 
    \sum_{i=0}^\infty \frac{\Big(\frac{C_{s,b}^2 (1-\eta_d)\eta_d^2 |\beta|^2}{(1-C_{t,b} t_s) (1-C_{t,b} t_s - C_{s,b} (1-\eta_d))}\Big)^i}{i!}
    \nonumber\\
    & \hspace{2.8cm} - \frac{e^{\frac{C_{s,b} \ \eta_d^2 |\beta|^2(1-\eta_s)}{1-C_{t,b} t_s(1-\eta_s)}} \ (1-\eta_s)}
    {\big(1-(C_{t,b} t_s+C_{s,b} (1-\eta_d)) (1-\eta_s)\big)^2}
    \sum_{i=0}^\infty \frac{\Big(\frac{C_{s,b}^2 (1-\eta_d) \eta_d^2 |\beta|^2 (1-\eta_s)^2}{(1-C_{t,b} t_s (1-\eta_s)) (1-(C_{t,b} t_s+C_{s,b} (1-\eta_d)) (1-\eta_s))}\Big)^i}{i!}\Bigg] 
    \nonumber\\
    & = \frac{C_{s,b}\eta_d\beta \, e^{-\eta_d |\beta|^2}}{\cosh^{2}{g_b}} 
    \Bigg[ \frac{\exp{\Big(\frac{C_{s,b} \ \eta_d^2 |\beta|^2}{1-C_{t,b} t_s - C_{s,b} (1-\eta_d)}\Big)}}
    {\big(1-C_{t,b} t_s - C_{s,b} (1-\eta_d)\big)^2}
    - \frac{\exp{\Big(\frac{C_{s,b} \ \eta_d^2 |\beta|^2 }{ 1-(C_{t,b} t_s+C_{s,b} (1-\eta_d)) (1-\eta_s)}\Big)}}
    {\big(1-(C_{t,b} t_s+C_{s,b} (1-\eta_d)) (1-\eta_s)\big)^2}\Bigg].
\end{align}

\noindent
For an SPPE initial state at Bob's station ($g_b\ll 1$), the joint probability takes the form
\begin{center}
    \fbox{%
    \parbox{0.9\textwidth}{%
    \begin{align}
        & \tilde{P}_\text{2ph}(00|\alpha,\beta) = 
        \frac{\exp{\Bigl(-\frac{(1-C)\eta_d|\alpha|^2}{1-C(1-\eta_d)}-\eta_d|\beta|^2 \Bigr)}
        \left({1-C}\right)^2}
        {(t_s + r_s \tanh^2{g})\ (1-C(1-\eta_d))^2}
        \times \nonumber\\
        &\Biggl[
        t_s\Bigl(1+ \frac{C\eta_d^2 |\alpha|^2}{1-C(1-\eta_d)} \Bigr) +
        r_s \tanh^2{g} \ \Bigl(1-\eta_d+ \frac{\eta_d^2 |\alpha|^2}{1-C(1-\eta_d)}\Bigr) 
        \bigl(1-\eta_d+\eta_d^2 |\beta|^2\bigr)
        + 2 \eta_d^2 \tanh{g} \sqrt{t_s r_s} Im(\alpha \beta )
        \Biggr] \nonumber
    \end{align}
    }%
    }
\end{center}

\section{Visibility between two SPDCs}
\label{vis-sec}

In this section, we analyze the effect of visibility between two SPDCs outputs on the CHSH value and the secure key rate (SKR) of the protocols.
Let $\hat{a}^\dagger_\omega$ be the creation operator of a photon with frequency $\omega$, satisfying the continuous-mode commutation relation $[\hat{a}_\omega, \hat{a}_{\omega'}^\dagger] = \delta(\omega - \omega')$.
Considering the spectral modes, the output of a SPDC generates spectrally correlated photons in the following form
\begin{align}
    \ket{\psi}_{\mathrm{SPDC}} = \sqrt{\lambda_{0}} \ket{00} + \sqrt{\lambda_{1}} \int \mathrm{d}\omega_{i} \mathrm{d}\omega_{s} F(\omega_{i}, \omega_{s}) a^{\dagger}_{\omega_{i}} a^{\dagger}_{\omega{s}} \ket{00} + \cdots,
\end{align}
where
$F(\omega_{i}, \omega_{s}) \propto \exp[-\frac{(\omega_{i} + \omega_{s} - 2\mu)^2}{2\sigma_{p}}] \mathrm{sinc}[\frac{L \Delta k}{2}]$ is the joint spectral amplitudes (JSA) \cite{branczyk2017hong} of the outputs, $\mathrm{sinc}(x) = \sin{x}/x$ and $L$ is the crystal length. In this relation, $\Delta k = k_{i}(\omega_{i}) + k_{s}(\omega_{s}) - k_{p}(\omega_{i} + \omega_{s})$ and $k$ is the wave number, $\sigma_{p}$ is bandwidth and $\mu$ is variance of pump pulse. 
Under specific conditions of $\frac{2}{\sigma^2} + \gamma L^2 (k'_{s} - k'_{p}) (k'_{i} - k'_{p}) = 0$ \cite{branczyk2017hong}, which occurs for SPDC crystals of a certain length, JSA functions become separable, and the state produced by SPDC can be written as
\begin{equation}
    \ket{\psi}_{SPDC} = 
    \sum_{n=0}^\infty \sqrt{\lambda_n}
    \Big[\frac{1}{\sqrt{n!}} \mathop{\scriptstyle \bigotimes}_{i=1}^{n} \int d\omega_{i} f_1(\omega_{i}) a_{\omega_{i}}^{\dagger} \ket{0}\Big]_1
    \otimes
    \Big[\frac{1}{\sqrt{n!}} \mathop{\scriptstyle \bigotimes}_{i=1}^{n}\int d\omega_{i} f_2(\omega_{i}) a_{\omega_{i}}^{\dagger} \ket{0}\Big]_2
    = \sum_{n=0}^\infty \sqrt{\lambda_n} \ket{n}_f \otimes \ket{n}_f
\end{equation}
where we defined the corresponding $n$-photon wavepacket state  characterized by $f(\omega)$ as $\ket{n}_f := \frac{(\hat{A}_f^\dagger)^n}{\sqrt{n!}} \ket{0}$ with the creation operator for a single photon in a wavepacket given by $\hat{A}_f^\dagger = \int d\omega\, f(\omega)\, \hat{a}^\dagger_\omega$.

\subsection{The 1-photon protocol}
We consider two SPDC sources at Alice’s and Bob’s stations, producing photons with separable JSA functions $f(\omega)$ and $g(\omega)$, respectively.
To evaluate Hong--Ou--Mandel (HOM) visibility at Charlie's, defined as $V:=|v|^2$, with visibility amplitude $v:=\int f^*(\omega) g(\omega)\mathrm{d}\omega$, the beams must interfere in a HOM interferometer~\cite{branczyk2017hong}. 
The projector for a detector that perfectly distinguishes single-photon events without revealing any frequency information about the photon is given by~\cite{branczyk2010optimized}:
\begin{equation}
    \Pi = \int \mathrm{d}\omega a^{\dagger}_{\omega} \ketbra{0}{0} a_{\omega}.
\end{equation}
Accordingly, Charlie's POVM, incorporating his 50:50 beam splitter, can be written as
\begin{equation}
    \Pi_{\text{BS}}^{(c)}
    = \frac{1}{2} \int_{-\infty}^{\infty} d\omega
    (a_{\omega}^{\dagger} - i b_{\omega}^{\dagger}) \ketbra{00}{00} (a_{\omega} + i b_{\omega}) 
    = \frac{1}{2} \int_{-\infty}^{\infty} d\omega
    (\ket{1_{\omega}0} -i\ket{01_{\omega}})(\bra{1_{\omega}0} +i\bra{01_{\omega}}),
\end{equation}
with $a_\omega (b_\omega)$ and $a^\dagger_\omega (b^\dagger_\omega)$ are the annihilation and creation operators, respectively, acting on the idler mode received from Alice (Bob).
To find the impact of the visibility of two beams arriving at Charlie's station on the final state shared between Alice and Bob, we begin by writing the state generated by their SPDC sources:
\begin{align}
\rho_{\text{in}} =  
\sum_{n,n',m,m' = 0}^{\infty} \sqrt{\lambda_{n} \lambda_{n'} \lambda_{m} \lambda_{m'}} &
\sum_{p=0}^{\min(n,n')}\sum_{q=0}^{\min(m,m')} \sqrt{{ n\choose p}{ n'\choose p} {m\choose q}{m'\choose q} (1-r)^{n+n'+m+m'-2(p+q)} r^{2(p+q)}} \nonumber \\
& \times \Big[\ketbra{n}{n'}_f \otimes \ketbra{m}{m'}_g \Big]_{a_1,b_1}
\times \Big[ \ketbra{n-p}{n'-p}_f \otimes \ketbra{m-q}{m'-q}_g \Big]_{a_2,b_2}.
\end{align}
where transmission losses are modeled by a BS with reflectivity $r$, and $p$ and $q$ indicate the number of photons lost in the idler beams. 
For simplicity, we focus on the idler beams and rewrite the state as below:
\begin{align}
    \rho_\text{in} = \sum_{n,n',m,m',p,q} \mathcal{F}^{(a_1,b_1)}_{m,n,m',n',p,q} ~
    \times \Big[ \ketbra{n-p}{n'-p}_f \otimes \ketbra{m-q}{m'-q}_g \Big]_{a_2,b_2}
\end{align}
with $\mathcal{F}^{(a_1,b_1)}_{m,n,m',n',p,q}$ containing all the coefficients and terms associated with the signal beams. 
To calculate the shared density matrix after Charlie's measurement, we need to apply his POVM on this state:
\begin{align}
    \rho_{\text{out}} & 
    = \mathrm{Tr}_{c_{1}c_{2}} \big(  \Pi_\text{BS}^{(c)} \ \rho_{in} \big) \nonumber \\
    & = \frac{1}{2} \sum_{n,n',m,m',p,q} \mathcal{F}^{(a_1,b_1)}_{m,n,m',n',p,q} ~ 
    \int d\omega \, \bra{00}
    (\alpha_{\omega} + i \beta_{\omega})  \ket{n-p}_f\ket{m-q}_g 
    \times \bra{n'-p}_f\bra{m'-q}_g (\alpha_{\omega'}^{\dagger} - i \beta_{\omega'}^{\dagger})  \ket{00} \nonumber\\
    & = \frac{1}{2} \sum_{n,n',m,m',p,q} \mathcal{F}^{(a_1,b_1)}_{m,n,m',n',p,q} ~ \int d\omega \Big( (f(\omega) \delta_{n,p+1} \delta_{m,q} + i g(\omega) \delta_{n,p} \delta_{m,q+1})(f^{*}(\omega) \delta_{n',p+1} \delta_{m',q} - i g^{*}(\omega) \delta_{n',p} \delta_{m',q+1} \Big) \nonumber \\
    &= \frac{1}{2} \sum_{n,n',m,m',p,q} \mathcal{F}^{(a_1,b_1)}_{m,n,m',n',p,q} ~ \Big( \delta_{n, n', p+1} \delta_{m, m', q} + \delta_{n, n', p} \delta_{m, m', q+1} + i \, v^{*} \delta_{n, n'-1, p} \delta_{m-1, m',q} - i \, v \, \delta_{n-1, n', p} \delta_{m, m'-1, q} \Big).
\end{align}
Applying the $\delta$-functions to  $\mathcal{F}^{(a_1,b_1)}_{m,n,m',n',p,q}$, reproduce the density matrix given in Eq.~\eqref{rho_out_1ph}, where $v$ appears exclusively in the off-diagonal components.  
Consequently, the visibility amplitude $v$ enters into the joint probabilities of Alice and Bob as follows, while leaving the marginal distributions unaffected:
\begin{align}
    {{P}^v_{1ph}(00|\alpha,\beta) 
    = \frac{(1 - C)^3}{2} \left( 
    \frac{2(1 - \eta)(1 - C(1 - \eta)) + \eta^2 \big(|\alpha|^2 + |\beta|^2 + 2\text{Im}(v\ \alpha\beta^*)\big)}{ (1-C(1-\eta))^4} \right)
    \exp\left( - \frac{ (1 - C)\eta (|\alpha|^2 + |\beta|^2)}{1-C(1-\eta)} \right).} \nonumber
\end{align}

\subsection{The 2-photon protocol}

For the 2-photon protocol, a similar analysis reveals that visibility affects the state in Eq.~\eqref{rho_out_2ph} in the same way, that is, $v$ appears exclusively in off-diagonal terms and consequently modifies only the joint probabilities as follows:
\begin{align}
    & \tilde{P}^v_{2ph}(00|\alpha,\beta) = 
    \frac{\exp{\Bigl(-\frac{(1-C)\eta|\alpha|^2}{1-C(1-\eta)}-\eta|\beta|^2 \Bigr)}
    \left({1-C}\right)^2}
    {(t_s + r_s \tanh^2{g})\ (1-C(1-\eta))^2}
    \times \nonumber\\
    &\Biggl[
    t_s\Bigl(1+ \frac{C\eta^2 |\alpha|^2}{1-C(1-\eta)} \Bigr) +
    r_s \tanh^2{g} \ \Bigl(1-\eta+ \frac{\eta^2 |\alpha|^2}{1-C(1-\eta)}\Bigr) 
    \bigl(1-\eta+\eta^2 |\beta|^2\bigr)
    + 2 \eta^2 \tanh{g} \sqrt{t_s r_s} Im(v \ \alpha \beta^{*} )
    \Biggr]. \nonumber
\end{align}

\section{Local Visibility between state and measurement}
\label{local-vis}

Here, we study the impact of local visibility between the photons in the quantum state and those involved in the measurement. 
Compared to our previous calculation, the only modification appears in the effective local measurement operators $\mathcal{M}^{(A)}_0$ and $\mathcal{M}^{(B)}_0$, and the shared state between Alice and Bob remains unchanged.
The marginals and joint probability distributions can be calculated using 
\begin{equation}\label{marg_vis}
        P(0|\delta)=Tr(\rho^{(x)}_{out} \, \mathcal{M}_0)=\sum_{n,n'=0}^\infty (\rho^{(x)}_{out})_{n'n} \bra{n}_{f_x}\mathcal{M}_0\ket{n'}_{f_x},
\end{equation}
\begin{equation}\label{joint_vis}
    P(00|\alpha\beta)
    = \mathrm{Tr}\!\left[\rho^{(ab)}_{\text{out}}
    \big(\mathcal{M}_0^{(A)} \otimes \mathcal{M}_0^{(B)}\big)\right]
    = \sum_{\substack{n,n'=0\\ m,m'=0}}^{\infty}
    \left(\rho^{(ab)}_{\text{out}}\right)_{\substack{n'm'\\ n\, m}}
    \, \bra{n}_{f_a} \mathcal{M}^{(A)}_0 \ket{n'}_{f_a}
    \, \bra{m}_{f_b} \mathcal{M}^{(B)}_0 \ket{m'}_{f_b} .
\end{equation}
where $\mathcal{M}_0 := D_g(\delta) E_{0}^{\eta_d} D^{\dagger}_g(\delta)$ consists of the no-click POVM $E_0^{\eta_d}=\sum_{m=0}^\infty (1-\eta_d)^m \, \hat{\Pi}_m$ in the presence of a displacement operator ${D}_g(\alpha)$, defined in a mode with normalized spectral amplitude $g(\omega)$.
We now compute the matrix elements of $\mathcal{M}_0$ in the Fock state basis $\ket{n}_f$, characterized by the normalized spectral amplitude $f(\omega)$.

The displacement operator ${D}_g (\alpha)$ acts on a specific mode with spectral amplitude $g(\omega)$, can be expressed as
\begin{equation}
{D}_g(\alpha) = \exp(\alpha \hat{A}_g^\dagger - \alpha^* \hat{A}_g)
\end{equation}

By utilizing the fact that any two normalized functions $f(\omega)$ and $g(\omega)$ in the Hilbert space $L^2(\mathbb{R})$ span a two-dimensional subspace, We decompose the function $f$ into a component parallel to the displacement mode $g$ and an orthogonal "dark" component $g^\perp$ via the Gram-Schmidt process:
\begin{equation}
f(\omega) = v \, g(\omega) + u \, g^\perp(\omega)
\Rightarrow 
\hat{A}_f^\dagger = v^* \hat{A}_g^\dagger + u  \hat{A}_{g^\perp}^\dagger
\end{equation}
where $v = \int f^*(\omega) g(\omega) d\omega$ is the spectral overlap integral and $|u|=\sqrt{1-|v|^2}$. Using this decomposition, the initial state $\ket{n}_f$ can be expanded in terms of the number of photons in these two orthogonal modes:
\begin{equation}
\ket{n}_f = \frac{(\hat{A}_f^\dagger)^n}{\sqrt{n!}} \ket{0} 
= \frac{(v \hat{A}_g^\dagger + u  \hat{A}_{g^\perp}^\dagger)^n}{\sqrt{n!}} \ket{0} 
= \sum_{p=0}^n \sqrt{\binom{n}{p}} v^{n-p} u^{p} \ket{n-p}_g \ket{p}_{g^\perp}.
\end{equation}
The displacement operator ${D}_g^\dagger(\alpha)$ acts exclusively on the $g$ mode, leaving the orthogonal mode $g^\perp$ unchanged:
\begin{equation} \label{D_g n_f}
    {D}_g^\dagger(\alpha)\ket{n}_f 
    = \sum_{p=0}^n \sqrt{\binom{n}{p}} v^{n-p} u^{p} \left( {D}_g^\dagger(\alpha)\ket{n-p}_g \right) \ket{p}_{g^\perp}.
\end{equation}

For a mode-blind detector, which clicks for any photon regardless of its frequency, the detection of $m$ photons mathematically defined as a projection operator $\hat{\Pi}_m$ onto the subspace where the sum of photons across all modes equals $m$:
\begin{equation}\label{projector to g}
\hat{\Pi}_m = \sum_{j=0}^m \ket{m-j}_g \ket{j}_{g^\perp} \bra{m-j}_g \bra{j}_{g^\perp}.
\end{equation}
Combining \ref{D_g n_f} and \ref{projector to g} yields us the matrix elements of $\mathcal{M}_0$ as
\begin{align}
    \bra{n}_f \mathcal{M}_0\ket{n'}_f 
    = \sum_{m=0}^{\infty} (1-\eta_d)^m 
    & \bra{n}_f D_g(\delta) \, \hat{\Pi}_m \, D_g^{\dagger}(\delta) \ket{n'}_f \nonumber\\
    =  \sum_{m=0}^{\infty} (1-\eta_d)^m
    & \Bigg[\sum_{p=0}^n \sqrt{\binom{n}{p}} (v^*)^{n-p} (u^*)^{p} \left( \bra{n-p}_g {D}_g(\alpha)\right) \bra{p}_{g^\perp}\Bigg] \nonumber \\
    \times & \Bigg[\sum_{j=0}^m \ket{m-j}_g \ket{j}_{g^\perp} \bra{m-j}_g \bra{j}_{g^\perp}\Bigg] \Bigg[\sum_{p'=0}^{n'} \sqrt{\binom{n'}{p'}} v^{n'-p'} u^{p'} \left( {D}_g^\dagger(\alpha)\ket{n'-p'}_g \right) \ket{p'}_{g^\perp}\Bigg] \nonumber\\
    = \sum_{m=0}^{\infty} (1-\eta_d)^m 
    & \hspace{-2mm}\sum_{j=0}^{min(m,n,n')} \hspace{-2mm}\sqrt{\binom{n}{j}\binom{n'}{j}} (v^*)^{n-j} v^{n'-j} |u|^{2j} \bra{n-j}_g {D}_g(\alpha)\ket{m-j}_g \bra{m-j}_g {D}^\dagger_g(\alpha)\ket{n'-j}_g,
\end{align}
where in the last line we used the identities $\braketm{p}{j}_{g^\perp}=\delta_{p,j}$ and $\braketm{p'}{j}_{g^\perp}=\delta_{p',j}$.
The single-mode matrix element $\bra{n-j} {D}(\alpha) \ket{m-j}$ can be evaluated using the normal-ordered expansion of the displacement operator, as derived in \ref{Displacement}:
\begin{equation}
    \bra{n-j}_g {D}_g(\alpha) \ket{m-j}_g 
    = e^{-\frac{|\alpha|^2}{2}} \sum_{k=j}^{\min(n,m)} \frac{\sqrt{(n-j)! (m-j)!} \ \alpha^{n-k} (-\alpha^*)^{m-k}}{(n-k)! (m-k)! (k-j)!}.
\end{equation}
By performing similar calculation as in \ref{Local Measurement Matrix}, we get
\begin{equation}\label{Local Measurement Matrix - vis}
    \bra{n}_f \mathcal{M}_0\ket{n'}_f 
    = e^{-\eta_d |\delta|^2} \sqrt{n! n'!} \hspace{2mm} \eta_d^{n+n'} (v^*\delta)^n (v\,\delta^*)^{n'} 
    \sum_{j=0}^{\min(n,n')} \frac{|\frac{u}{v}|^{2j}}{j!}
    \sum_{i=j}^{\min(n,n')} \frac{(\frac{1-\eta_d}{\eta_d^2 |\delta|^2})^{i}}{i!(n-i)!(n'-i)!}.
\end{equation}
In the limit of perfect matching ($V=1$), the only non-zero term corresponds to $j=0$, for which the above expression reduces to the \ref{Local Measurement Matrix}. 
To find the probability distributions, one needs to substitute \ref{Local Measurement Matrix - vis} into \ref{marg_vis} and \ref{joint_vis}.

Fig.~\ref{fig:vis-det} illustrates how the required visibility threshold varies with detection efficiency for the two protocols. Specifically, panel (a) shows the required visibility to violate the CHSH inequality at each detection efficiency, while panel (b) shows the thresholds needed to achieve a positive SKR for both the 1-photon and 2-photon cases.

\begin{figure}[H]\centering 
    \includegraphics[width=16cm]{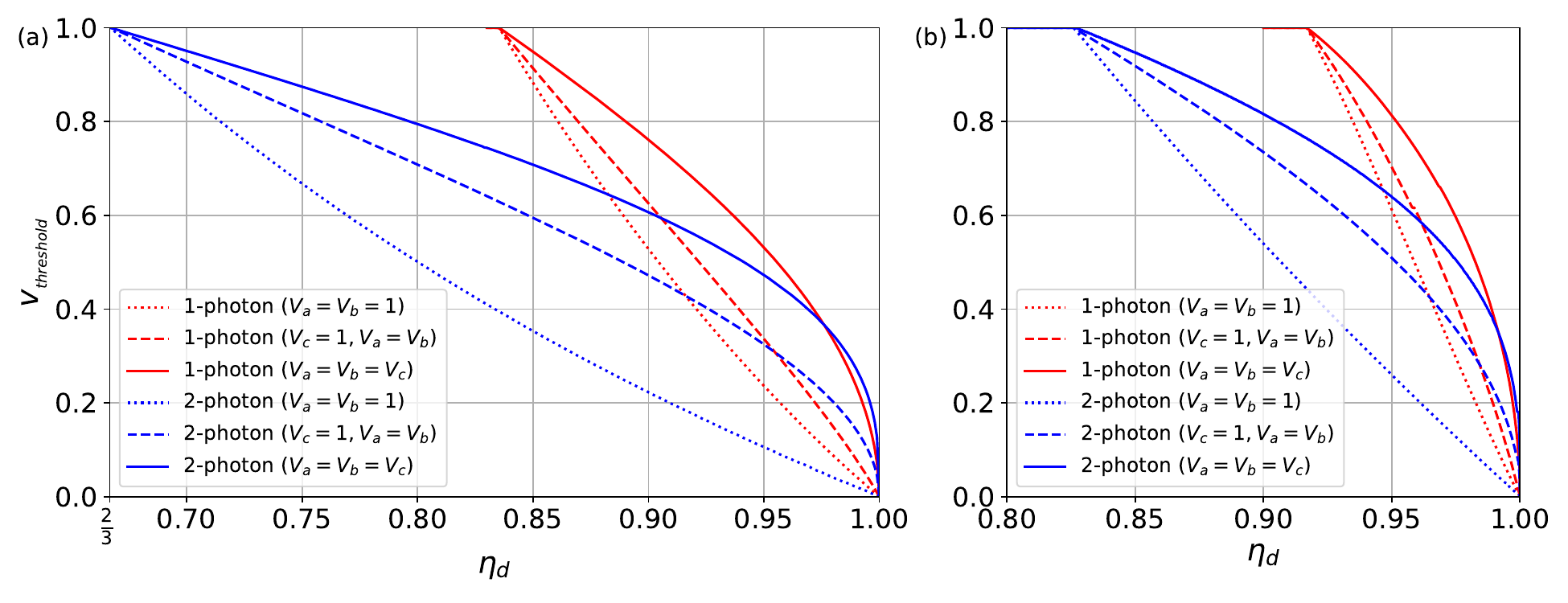}
    \caption{Required visibility (a) to violate the CHSH inequality and (b) to achieve a positive SKR as a function of detection efficiency, shown for the 1-photon (red) and 2-photon (blue) protocols. Solid lines correspond to imperfect visibility in all labs ($V_a=V_b=V_c<1$), while the dashed line assumed the perfect visibility in Charlie's lab ($V_c=1$), and in the dotted line correspond to perfect visibility in the local labs ($V_a=V_b=1$)}
    \label{fig:vis-det}
\end{figure}

\section{Lower bounds on the key rate}
\label{LB_Rate-sec}

The asymptotic key rate is determined using the Devetak–Winter formula, $r_{DW}=H(A_1|E)-H(A_1|B_3)$, where the first term quantifies Eve’s uncertainty about Alice’s outcomes, and the second term represents the information leaked during error correction. In our analysis, we employ three distinct computational techniques to obtain a rigorous lower bound on it, and complement these with two post-processing strategies designed to enhance the final value:

\subsection{Analytical approach: CHSH-based method}
The asymptotic secure key rate can be lower-bounded analytically using 
\begin{equation} \label{Analytical_LB}
    r \geq 1 - h\left(\tfrac{1+\sqrt{(S/2)^2-1}}{2}\right)- H(A_1|B_3)
    + h\left(\tfrac{1+\sqrt{1-q(1-q)(8-S^2)}}{2}\right),
\end{equation}
where $q$ is the probability that Alice flips her output, a procedure known as noisy preprocessing. 
Using this lower bound, Fig.~\ref {fig:2ph-Eta84to90}(a) illustrates the key-rate dependence on distance for the two-photon protocol with perfect visibility $V=1$, optimized over $g_a$ for various values of efficiencies, assuming $\eta_s=1$ and $g_b=0.1$.
The results are obtained using the probability distribution from Sec.~\ref{App:Imperfect-lossy-2ph} and the heralding efficiency from Eq.~\eqref{eq:2ph heralding proba}.
The corresponding optimal values of $g$ for each distance and efficiency are shown in Fig.~\ref{fig:2ph-Eta84to90}(b).

\begin{figure}[h]\centering
    \includegraphics[width=\columnwidth]{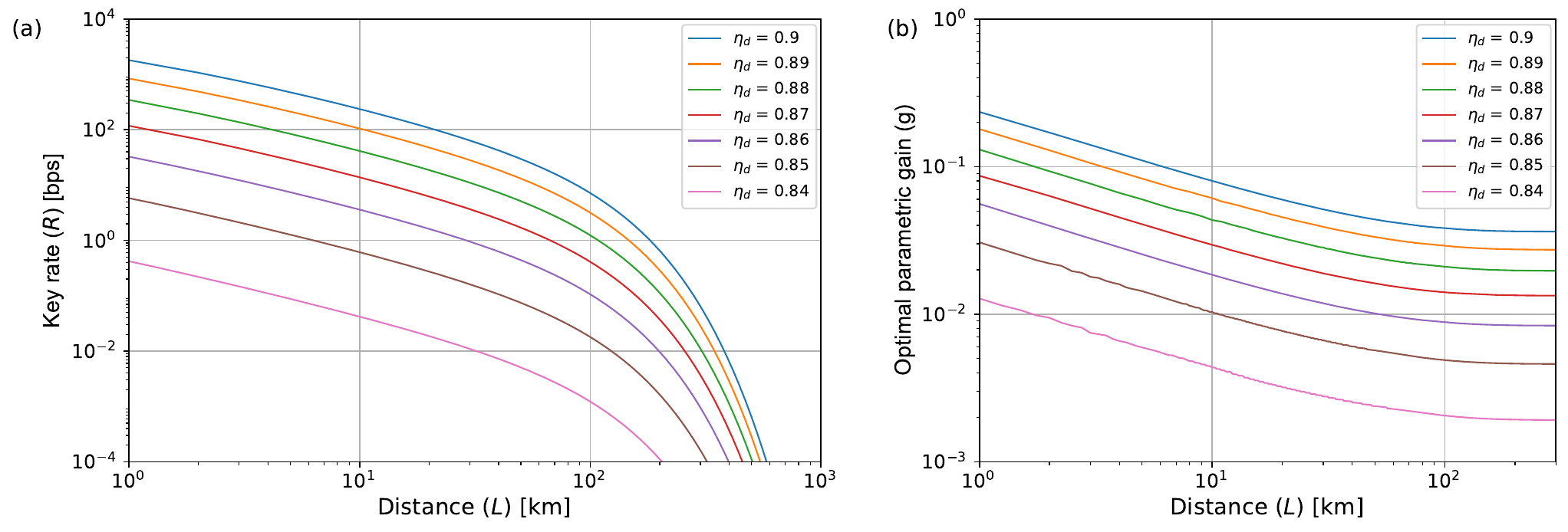}
    \caption{(a) Asymptotic secure key rate (bits per second, bps) versus distance for the 2-photon protocol using the optimal parametric gain $g_a$ with visibility $V=1$, assuming laser repetition rate $f_{\text{rep}}=100$ MHz, $\eta_s=1$, and $g_b=0.1$, using the bound from Eq.~\eqref{SKR-CHSH+NP}.
    (b) Optimal value of $g_a$ as a function of distance for the 2-photon protocol with visibility $V=1$.
    Lines correspond to detection efficiencies from 84\% to 90\%.}
    \label{fig:2ph-Eta84to90}
\end{figure}

\subsection{Numerical approach: Brown--Fawzi--Fawzi (BFF) method}
\label{BFF-sec}

The results of~\cite{Brown2024} enable the construction of a numerical upper bound on the relative entropy, approximated by a polynomial. Using this, a lower bound on the conditional von Neumann entropy of the state $\rho$ is established based on its connection with the relative entropy $D(\rho \Vert \sigma) = \mathrm{Tr}[\rho(\log \rho - \log \sigma)]$ in the following way
\begin{equation}
    H(A \vert B)_{\rho} := -D(\rho_{AB} \Vert \mathbb{I}_{A}\otimes \rho_{B}).
\end{equation}

Subsequently, NPA hierarchy is employed to transform the optimization problem of this non-commutative polynomial into a SDP, and the lower bound for the conditional entropy of Alice and Eve is calculated with the following program

\begin{align}    \label{Brown-eq}
    H(A|X=x^*, \mathcal{H}_{E}) &\ge \sum_{i=1}^{m-1} \frac{\omega_i}{t_i \ln2} + \inf \sum_{i=1}^{m-1} \frac{\omega_i}{t_i \ln2} \sum_{a} \bra{\psi} M_{a|x^*} \left( (Z_{a, i} + Z_{a, i}^{*} +(1-t_{i}) Z_{a, i}^{*} Z_{a, i} \right) + t_{i} Z_{a, i}Z_{a, i}^{*} \ket{\psi} , \\
    s.t. ~ \sum_{abxy} c_{abxyj} & \bra{\psi} M_{a|x} N_{b|y} \ket{\psi} \geq v_{j}, \quad \text{for all}~ 1 \leq j \leq r \nonumber \\
    \sum_{a} M_{a|x} &= \sum_{b} N_{b|y} = \mathcal{I}, \quad \text{for all}~ x, y \nonumber \\
    M_{a|x} &\geq 0,\quad \text{for all}~ a, x \nonumber \\
    N_{b|y} &\geq 0, \quad \text{for all}~ b, y \nonumber \\
    Z_{a, i}^{*} Z_{a, i} &\leq \alpha_{i},  \quad \text{for all} ~a, i = 1, \cdots , m-1 \nonumber \\
    Z_{a, i} Z_{a, i}^{*} &\leq \alpha_{i},  \quad \text{for all} ~a, i = 1, \cdots , m-1 \nonumber \\
    [M_{a|x},N_{b|Y}]& = [M_{a|x},Z_{b, i}^{*}] = [N_{b|y},Z_{a, i}^{*}] = 0 , \quad \text{for all}~ a, b, x, y, i \nonumber\\
    M_{a|x}, N_{b|y}, Z_{z, i} &\in B(H), \quad \text{for all} ~ a, b, x, y, i \nonumber
\end{align}
Note that considering pure states is enough for this optimization. In the above relation, $\omega_{i}$s are the weights in Gauss--Radau quadrature, $Z_{i}$s are Eve's bounded operators, and $\{\{M_{a|x}\}_{a}\}_{x}$ are preselected POVM's used by Alice. Furthermore,  $B(\mathcal{H}_{E})$ is von Neumann algebra on Eve's separable Hilbert space $\mathcal{H}_E$. For more details refer to \cite{Brown2024}. 

In our numerical calculations, we go up to the second level of NPA hierarchy plus some additional constraints of higher level ``$\mathrm{ABZ} + \mathrm{AZZ} + \mathrm{ABB} + \mathrm{AAB}$'' where $A \in \{I\} \cup \{M_{a|x}\}_{a,x}$, $B \in \{I\} \cup \{N_{b|y}\}_{b,y}$, and $Z \in \{I\} \cup \{Z_{c,i}, Z_{c,i}^* \}_{c,i}$. We do not go to further levels due to computational complexity.

To combine noisy preprocessing with this method, we should modify the measurement operators in Eq.~\eqref{Brown-eq} to $\hat{M}_{a|x} = (1-q) M_{a|x} + q M_{a|x}$. In addition, the first term in the RHS of this relation should be replaced with $c_m = 2q(1-q) + \sum_{i = 1}^{m-1} \frac{\omega_{i}}{t_i \ln{2}}$. For detailed calculations, see Ref.~\cite{Brown2024}.

\subsection{Numerical approach: Min entropy ($H_{min}$) method}
\label{gussing-sec}

Here, we present a numerical approach,  based on min-entropy as a lower bound for the conditional von Neumann entropy~\cite{Pan-PRL}, to calculate the key rate. Although this method may not yield better results compared to the Brown-Fawzi-Fawzi (BFF) method, it offers a computational advantage. Besides, its combination with the post-selection method leads to an improvement in the efficiency threshold. Post-selection is a technique that is used to overcome errors caused by detection deficiencies in which Alice and Bob randomly and independently keep bits `0' with probability $p$, discard them with probability $1-p$, and keep all bits `1' intact during their key generation rounds.

The min-entropy $H_{min}$ is tied to guessing probability $G(A_1 \vert E, \nu_p)$ in the following way
\begin{equation}
    H_{\text{min}}(A_{1}|E, \nu_{p}) = - \log_2 G(A_{1}|E, \nu_{p}),
\end{equation}
where $\nu_p = \{ab ~\vert ~ ab = 00, 01, 10, 11 \}$ is the set of postselected events with coefficients $\omega_{00}=1$, $\omega_{10}=\omega_{01}=p$, and $\omega_{11}=p^2$, and the guessing probability is defined as
\begin{equation}
    G(A_1 \vert E, \nu_p) = \frac{1}{p_{\nu_p}} \max\limits_{P(a, b, e \vert 1, 3, z)} \sum_{a,b \in \nu_p} \omega_{ab} P(a, b, a \vert 1, 3, z),
\end{equation}
with $p_{\nu_p} = \sum_{a,b\in \nu_{p}} \omega_{ab} P(a, b\vert 1, 3)$ is the total probability of keeping a pair of bits.
This optimization can be realized by means of SDP method, with the following constraints
\begin{align*}
    \sum_{e \in \{0,1\}} P(a, b, e\vert 1, 3, z) &{} = P(a, b \vert 1, 3),\\
    P(a, b, e\vert 1, 3, z) &{}\in \Tilde{Q}.
\end{align*}
The second constraint ensures that the conditional probabilities remain within the quantum set. To verify this, we utilize the NPA hierarchy \cite{Mironowicz, NPA}. 
The optimized lower bound on key rate in this case is calculated using the following modification of Eq.~\eqref{DW-ineq}
\begin{equation}
    r \ge p_{\nu_p} [H_{\text{min}}(A_1|E, \nu_{p}) - H(A_1|B_3, \nu_{p})].
\end{equation}
 The technique improves the threshold efficiency to 89.8\% in the 1-photon protocol, but compromises security by restricting it to collective attacks rather than the coherent attacks.

\subsection{Comparison}

In this section, we compare the previously described lower bound approaches and present numerical values for the parameters introduced in the manuscript. 

Fig.~\ref{fig:keyrate methods-Comparision} illustrates the maximum achievable key rate for both protocols as a function of detection efficiency.
Moreover, Table~\ref {table:det-threshold} compares the detection efficiency required to obtain a positive key rate and Bell violation for the lower bound methods discussed in this section.

\begin{figure}[H]\centering
    \includegraphics[width=\linewidth]{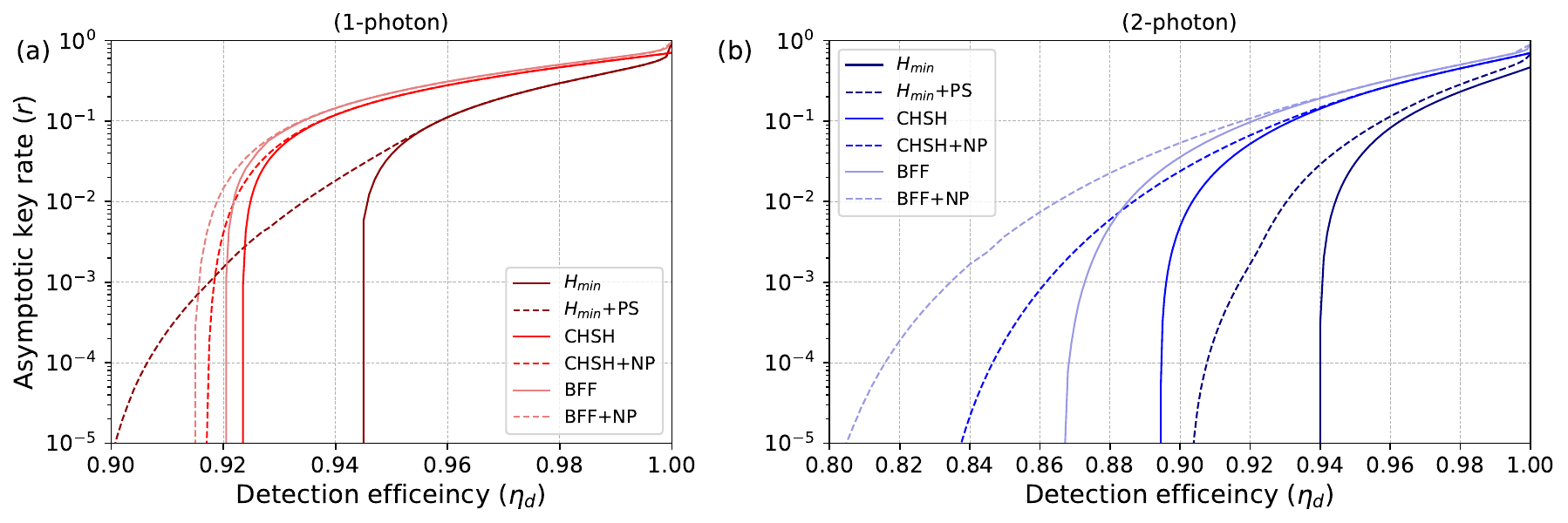}
    \caption{Secure key rate computed using CHSH, BFF, $H_{min}$ methods, without (solid) and with postprocessing (dashed) using Noisy Preprocessing (NP) and Post Selection (PS) techniques for (a) the 1-photon protocol, and (b) the 2-photon protocol.}
    \label{fig:keyrate methods-Comparision}
\end{figure}

\begin{table}[H]
	\centering
	\begin{tabular}{|c|c|c|c|}
		\hline
		Method &  related  & \multicolumn{2}{|c|}{min $\eta_{d}$} \\
		\cline{3-4}
		 & parameters  & 1-photon & 2-photon \\
		\hline
		&  Bell violation & 0.836 & 2/3  \\		 
		\cline{2-4}
		CHSH & r & 0.924 & 0.895 \\
		\cline{2-4}
		method &  r + NP & 0.917 & 0.826 \\		
		\hline
        
		\hline
		BFF & r & 0.921 & 0.868 \\
		\cline{2-4}
		method & r + NP & 0.915 & 0.802 \\
		\hline
        
		\hline
	    Min & r & 0.945 & 0.94 \\
		\cline{2-4}
		  entropy & r + PS & 0.898 & 0.903 \\		
		\hline
	\end{tabular}
	\caption{Detection efficiency threshold required to violate CHSH inequality and to achieve a positive secure key rate. For the Noisy Preprocessing (NP) technique, the optimal value of $q$ rises from 0 to 0.5 as the efficiency decreases from 1 toward the threshold, whereas for the Post Selection (PS) technique, the optimal value of $p$ drops from 1 to 0 over the same range.}
	\label{table:det-threshold}
\end{table}

\vspace{-8mm}
\section{Finite-key rate}
\label{finite-sec}
To optimize the key rate in the finite-size regime, we developed a three-tier security certification strategy:
\begin{enumerate}

\item \textit{CHSH certificate}: For each distance $L$, We optimize the Devetak--Winter bound using the analytical CHSH formula with noisy preprocessing. This allowed us to identify optimal optical parameters and bit-flip probability $q_{\text{opt}}$, from which the expected probability distribution $P(a,b\vert x,y)_L$ is obtained. Using the CHSH inequality as the Bell certificate, we then apply the BFF method to construct a min-tradeoff function and evaluate the finite-size key rate via the EAT.

\item \textit{Full statistic certificate}: Using $P(a,b|x,y)_L$, we apply the full statistics approach~\cite{Nieto-Silleras2013,Bancal2014,Brown2024}, which puts constrains on the complete probability distribution rather than only the CHSH value. The resulting min-tradeoff functions $f_L$, were then
used in conjunction with the EAT to compute finite-size key
rates. This method provides the tightest security bounds but requires monitoring of the whole measurement statistics.

\item \textit{Optimal Bell certificate}:  To overcome the real-time implementation challenges of method (2), we define a new distance-dependent Bell inequality using the min-tradeoff function $f_L$ derived from full-statistics optimization. We computed the observed value of this optimal Bell certificate for $P(a,b|x,y)_L$, and enforced this value as a constraint in a subsequent BFF procedure, with the final key rate obtained again through the EAT. This custom certificate achieves key rates that significantly exceed CHSH-based bounds while maintaining practical implementability.

\end{enumerate}
The effect of the three certification methods on the finite-key rate is illustrated in Fig.~\ref{fig:Finite_(Full,Bell,CHSH)}. The optimized Bell certificate consistently outperforms CHSH-based certification, delivering 2--3 times higher key rates at maximum distances. Full statistics provide marginal additional gains at a substantial implementation cost. Notably, numerical optimization yields only modest improvements over the analytical bound, Eq.~\eqref{SKR-CHSH+NP}, reflecting near-optimality of CHSH bounds for maximally entangled states. This simplifies experimental implementation, as monitoring the CHSH value $S$ alone suffices for accurate key rate estimation.

\begin{figure}[htbp]
    \centering
    \includegraphics[width=0.9\linewidth]{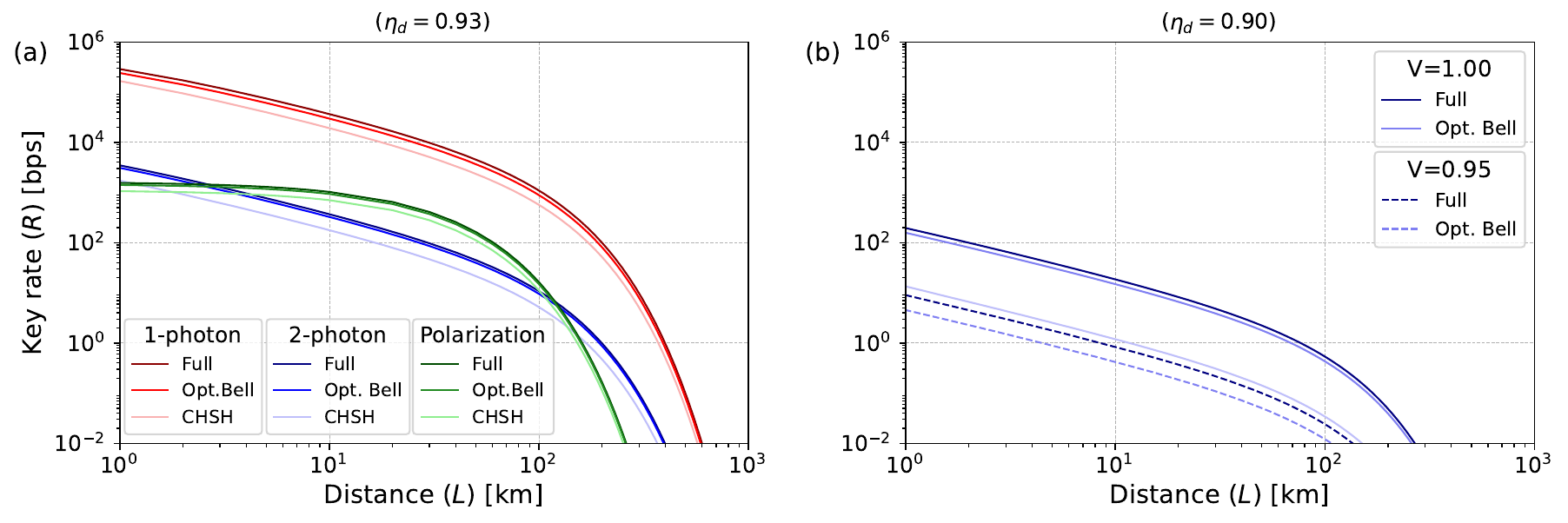}
    \caption{{Comparison of finite-size and asymptotic key rates $R$ versus distance $L$  under different certification approaches, assuming $N = {10^9}$ protocol rounds, repetition rate $f_{\text{rep}}=100$ MHz, optimal parametric gain $g$, computed for:
    (a) local detection efficiency $\eta_d = 93\%$ for the 1-photon protocol (red), the 2-photon protocol (blue), and the polarization-based protocol from Ref.~\cite{Oudot2024} (green), all for visibility $V=1$, 
    (b) local detection efficiency $\eta_d = 90\%$ for the 2-photon protocol, computed for visibilities $V=0.95$ (dashed lines) and 1 (solid lines).}}
    \label{fig:Finite_(Full,Bell,CHSH)}
\end{figure}

\end{document}